\documentclass{amsart}
\usepackage{mathrsfs}
\usepackage{graphicx}
\usepackage{amsmath}
\usepackage{amscd}
\usepackage{cite}
\usepackage{hyperref}
\usepackage{epsfig}
\usepackage{subfigure}
\usepackage{caption}
\usepackage{tikz}
  \usepackage{float}
  \usepackage{amssymb}
\usepackage{amsfonts}
\usepackage[all,cmtip]{xy}
\usepackage{appendix}
\usepackage{bm}

\newtheorem{Example}{Example}[section]

\newtheorem{Theorem}{Theorem}[section]
\newtheorem{Theorem/Definition}{Theorem/Definition}[section]
\newtheorem{Proposition}{Proposition}[section]
\newtheorem{Lemma}{Lemma}[section]

\newcommand{\pd}{\partial}

\newcommand{\bC}{{\mathbb C}}

\newcommand{\bP}{{\mathbb P}}

\newcommand{\bZ}{{\mathbb Z}}

\newcommand{\cF}{{\mathcal F}}

\newcommand{\cH}{{\mathcal H}}

\newcommand{\cM}{{\mathcal M}}

\newcommand{\half}{\frac{1}{2}}

\newcommand{\Mbar}{\overline{\cM}}

\newcommand{\wA}{{\widehat A}}
\newcommand{\wB}{{\widehat B}}

\newcommand{\be}{\begin{equation}}
\newcommand{\ee}{\end{equation}}
\newcommand{\bea}{\begin{eqnarray}}
\newcommand{\ben}{\begin{eqnarray*}}
\newcommand{\een}{\end{eqnarray*}}
\newcommand{\eea}{\end{eqnarray}}

\DeclareMathOperator{\Pf}{Pf}

\usepackage{color}

\definecolor{yellow}{rgb}{1,1,0}
\definecolor{orange}{rgb}{1,.7,0}
\definecolor{red}{rgb}{1,0,0}
\definecolor{green}{rgb}{0,1,1}
\definecolor{white}{rgb}{1,1,1}

\definecolor{A}{rgb}{.75,1,.75}

\theoremstyle{remark}
\newtheorem{Remark}{Remark}[section]

\begin{document}

\newtheorem{myDef}{Definition}
\newtheorem{thm}{Theorem}
\newtheorem{eqn}{equation}

\title[BKP hierarchy and connected bosonic $N$-point functions]
{BKP Hierarchy, Affine Coordinates, and a Formula for Connected Bosonic $N$-Point Functions}

\author{Zhiyuan Wang}
\address{School of Mathematical Sciences\\
Peking University\\Beijing, 100871, China}
\email{zhiyuan19@math.pku.edu.cn}

\author{Chenglang Yang}
\address{Beijing International Center for Mathematical Research\\
Peking University\\Beijing, 100871, China}
\email{yangcl@pku.edu.cn}

\begin{abstract}

We derive a formula for the connected $n$-point functions of a tau-function
of the BKP hierarchy in terms of its affine coordinates.
This is a BKP-analogue of a formula for KP tau-functions
proved by Zhou in \cite{zhou1}.
Moreover,
we prove a simple relation between the KP-affine coordinates of a tau-function $\tau(\bm t)$ of the KdV hierarchy
and the BKP-affine coordinates of $\tau(\bm t/2)$.
As applications,
we present a new algorithm to compute the free energies of
the Witten-Kontsevich tau-function and the Br\'ezin-Gross-Witten tau-function.

\end{abstract}

\maketitle

\tableofcontents

\section{Introduction}

Integrable systems have drawn a lot of attention in mirror symmetry
since the Witten Conjecture/Kontsevich Theorem \cite{wit,kon}.
The boson-fermion correspondence developed by Kyoto School
is one of the most interesting approaches to study integrable hierarchies
such as the KP (Kadomtsev-Petviashvili) hierarchy,
KdV (Korteweg-de Vries) hierarchy, BKP hierarchy, etc,
since it establishes a connection to representation theory and symmetric functions.
See \cite{djm} for an introduction to Kyoto School's approach to the KP hierarchy and Sato's theory.

In Kyoto School's approach,
a tau-function can be regarded as either a vector in the bosonic Fock space,
or a vector in the fermionic Fock space,
satisfying the bosonic or fermionic Hirota bilinear relations respectively.
Moreover,
Sato found that the space of all tau-functions of the KP hierarchy
is a semi-infinite dimensional Grassmannian \cite{sa}.
See also \cite{sw} for an analytic construction.
This Grassmannian is the orbit of the trivial tau-function $\tau= 1$ under
the action of an infinite-dimensional group $\widehat{GL(\infty)}$.
A traditional way to express a tau-function in the fermionic picture is
$\tau = e^g |0\rangle$,
where $|0\rangle$ is fermionic vacuum and $g\in \widehat{\mathfrak{gl}(\infty)}$
is of the form:
\begin{equation*}
g=\sum_{m,n\in \bZ} c_{n,m}\psi_{-m-\half} \psi_{-n-\half}^*,
\end{equation*}
such that $c_{n,m}=0$ for $|n+m|>>0$.
Here $\{\psi_r,\psi_s\}_{r,s\in \bZ+\half}$ are the fermions
(and we choose the notations such that $\psi_r,\psi_r^*$ are fermionic creators when $r<0$).

In literatures
there is an alternative way to express a tau-function in the fermionic space.
A tau-function with $\tau(0)=1$ can be uniquely represented as a Bogoliubov transform of the vacuum
which only involves fermionic creators (see e.g. \cite[\S 3]{zhou1}):
\be
\label{eq-intro-Bog}
\tau = \exp \big( \sum_{m,n\geq 0} a_{n,m} \psi_{-m-\half}\psi_{-n-\half}^* \big) |0\rangle.
\ee
If one applys the boson-fermion correspondence and takes KP-time variables to be
$T_n = \frac{p_n}{n}$ where $p_n$ is the Newton symmetric function of degree $n$,
then:
\begin{equation*}
\tau = \sum_\mu (-1)^{n_1+\cdots+n_k}\cdot \det(a_{n_i,m_j})_{1\leq i,j\leq k}
\cdot s_\mu,
\end{equation*}
where $\mu=(m_1,\cdots,m_k|n_1,\cdots,n_k)$ is a partition of integer
(written in the Frobenius notation),
and $s_\mu$ is the Schur function indexed by $\mu$.
See \cite{adkmv, dz, zhou3} for examples of representing tau-functions as
Bogoliubov transforms of the form \eqref{eq-intro-Bog}.
The coefficients $\{a_{n,m}\}$ are called the affine coordinates of $\tau$,
and they provide a canonical choice of coordinates on the big-cell of the Sato Grassmannian (see e.g. \cite{by}).

A natural question is,
how to compute the logarithm $\log\tau$ (called the free energy)
of a tau-function $\tau$
using its affine coordinates $\{a_{n,m}\}$.
This is crucial since the coefficients of some free energies
are important invariants in geometry.
For example,
the coefficients of the free energy associated to the Witten-Kontsevich tau-function \cite{wit, kon}
are the intersection numbers of $\psi$-classes on the moduli spaces $\Mbar_{g,n}$ of stable curves.
In \cite{zhou1},
Zhou has proved the following formula for the connected $n$-point functions associated
to a KP tau-function
(see \cite[Theorem 5.3]{zhou1}):
\be
\label{eq-intro-Zhou}
\begin{split}
&\sum_{i_1,\cdots,i_n \geq 1}
\frac{\pd^n \log\tau(\bm T)}{\pd T_{i_1}\cdots \pd T_{i_n}}\bigg|_{\bm T=0}
\cdot z_1^{-i_1-1} \cdots z_n^{-i_n-1}\\
=& (-1)^{n-1}\cdot \sum_{\sigma:\text{ $n$-cycles}}
\prod_{i=1}^n \wA(z_{\sigma(i)},z_{\sigma(i+1)})
-\frac{\delta_{n,2}}{(z_1-z_2)^2},
\end{split}
\ee
where $\bm T=(T_1,T_2,\cdots)$ are the KP-time variables
and $\wA(w,z)$ is the generating series of the affine coordinates $\{a_{n,m}\}_{n,m\geq 0}$.
Moreover,
he found a formula for the generating series of affine coordinates
of the Witten-Kontsevich tau-function (see \cite[\S 6.9]{zhou1}),
thus one can indeed carry out the calculations of the intersection numbers on $\Mbar_{g,n}$
using \eqref{eq-intro-Zhou}.
See also \cite{zhou2, zhou4, wa, wz} for some applications of this formula
to other well-known KP tau-functions.

The goal of the present work is
to find a BKP-version of \eqref{eq-intro-Zhou}.
BKP hierarchy is an integrable system introduced by Kyoto School \cite{djkm, jm},
which shares a lot of common properties with the KP hierarchy.
In particular,
one also has the fermionic description (in terms of the neutral fermions)
and semi-infinite dimensional Grassmannian (the isotropic Grassmannian)
description of a BKP tau-function.
See e.g. \cite{yo, lw, or, tu, va, kv} for more information about the BKP hierarchy
and BKP tau-functions,
and see \cite[\S 7]{hb} for an introduction to BKP affine coordinates.
We will give a brief review of these notions in \S \ref{sec-pre}.

Let $\tau=\tau(\bm t)$ be a BKP tau-function with $\tau(0)=1$,
where $\bm t = (t_1,t_3,t_5,\cdots)$.
then $\tau$ can be represented as
a Bogoliubov transform in the fermionic Fock space:
\begin{equation*}
\tau = \exp\big(\sum_{n,m\geq 0} a_{n,m} \phi_m \phi_n \big) |0\rangle,
\end{equation*}
where $\{\phi_i\}_{i\geq 0}$ are the neutral fermionic creators.
The coefficients $\{a_{n,m}\}_{n,m\geq 0}$ are the BKP-affine coordinates of $\tau$
(satisfying the condition $a_{n,m}=-a_{m,n}$).
In the bosonic Fock space,
$\tau$ is a summation of Schur Q-functions:
\begin{equation*}
\tau =
\sum_{\mu \in DP} (-1)^{\tilde l(\mu)/2} \cdot
\Pf (a_{\mu_i,\mu_j})_{1\leq i,j\leq \tilde l(\mu)} \cdot Q_\mu(\bm x).
\end{equation*}
Now denote by $A(w,z)$ and $\wA(w,z)$ the following generating series respectively:
\begin{equation*}
\begin{split}
&A(w,z) = \sum_{n,m>0} (-1)^{m+n+1} \cdot a_{n,m} w^{-n} z^{-m}
-\half \sum_{n>0}(-1)^n a_{n,0} (w^{-n}-z^{-n}),\\
&\wA(w,z) =  A(w,z)-\frac{1}{4} -\half\sum_{i=1}^\infty (-1)^{i} w^{-i} z^i.
\end{split}
\end{equation*}
They are actually the fermionic two-point functions
(see \S \ref{sec-npt-disconn}).
Our main result is the following formula for the connected bosonic $n$-point functions
(see \S \ref{sec-conn-npt}):
\begin{Theorem}
\label{thm-intro-formula}
Let $\tau$ be a BKP tau-function satisfying $\tau(0)=1$,
and let $A,\wA$ be the generating series of its affine coordinates defined as above.
Then:
\begin{equation*}
\sum_{i> 0: \text{ odd}}
\frac{\pd \log\tau(\bm t)}{\pd t_{i}} \bigg|_{\bm t=0}
\cdot z^{-i}
=A(-z,z),
\end{equation*}
and for $n\geq 2$,
\begin{equation*}
\begin{split}
&\sum_{i_1,\cdots,i_n> 0: \text{ odd}}
\frac{\pd^n \log\tau(\bm t)}{\pd t_{i_1}\cdots \pd t_{i_n}} \bigg|_{\bm t=0}
\cdot z_1^{-i_1}\cdots z_n^{-i_n}
=
-\delta_{n,2} \cdot
\frac{z_1z_2(z_2^2+z_1^2)}{2(z_1^2-z_2^2)^2}\\
&\qquad\qquad\qquad
+ \sum_{\substack{ \sigma: \text{ $n$-cycle} \\ \epsilon_2,\cdots,\epsilon_n \in\{\pm 1\}}}
(-\epsilon_2\cdots\epsilon_n) \cdot
\prod_{i=1}^n \xi(\epsilon_{\sigma(i)} z_{\sigma(i)}, -\epsilon_{\sigma(i+1)} z_{\sigma(i+1)}),
\end{split}
\end{equation*}
where $\xi$ is given by:
\begin{equation*}
\begin{split}
\xi(\epsilon_{\sigma(i)} z_{\sigma(i)}, -\epsilon_{\sigma(i+1)} z_{\sigma(i+1)}) =
\begin{cases}
\wA (\epsilon_{\sigma(i)} z_{\sigma(i)}, -\epsilon_{\sigma(i+1)} z_{\sigma(i+1)}),
& \sigma(i)<\sigma(i+1);\\
-\wA( -\epsilon_{\sigma(i+1)} z_{\sigma(i+1)} ,\epsilon_{\sigma(i)} z_{\sigma(i)}),
& \sigma(i)>\sigma(i+1),
\end{cases}
\end{split}
\end{equation*}
and we use the conventions
$\epsilon_{1} :=1$ and
$\sigma(n+1):=\sigma(1)$.
\end{Theorem}

Furthermore,
given a tau-function $\tau(\bm t)$ of the KdV hierarchy
(see \cite{djm} for an introduction to KdV),
one knows that $\tau(\bm t/2)$ is a tau-function of the BKP hierarchy \cite{al1}.
We find the following (see \S \ref{sec-kdv} for details):
\begin{Theorem}
Let $\tau(\bm t)$ be a tau-function of the KdV hierarchy.
Then:
\be
\label{eq-intro-KPBKP}
A^{\text{BKP}} (w,z) = -\frac{w-z}{4} \cdot A^{\text{KP}}(w,-z),
\ee
where $A^{\text{BKP}}(w,z)$ is the generating series of the BKP-affine coordinates of $\tau(\bm t/2)$,
and $A^{\text{KP}}(w,-z)$ is the generating series (introduced by Zhou in \cite{zhou1})
of the KP-affine coordinates of $\tau(\bm t)$.
\end{Theorem}

We will discuss some applications of the above formulas.
The Witten-Kontsevich tau-function $\tau_{\text{WK}}$ \cite{wit, kon}
and the Br\'ezin-Gross-Witten (BGW) tau-function $\tau_{\text{BGW}}$ \cite{bg, gw}
are two well-known tau-functions of the KdV hierarchy.
In literatures,
there have been already various methods to compute their free energies,
see e.g. \cite{lx, dvv, fkn} and \cite{al4, gn} respectively.
However,
there are still mathematical aspects
which have not been fully understood yet,
and we hope
the discussions in this work may provide some new understandings
from the point of view of the BKP hierarchy.
This is inspired by the works \cite{mm, ly2, ly1, al2, al3},
in which these two tau-functions were related to Schur Q-functions;
and the work of Zhou \cite{zhou1},
in which the KP-affine coordinates and boson-fermion correspondence
are used to compute the free energies.

Using the results in \cite{mm, al2, ly1, ly2, al3},
we are able to write down the explicit expressions of the BKP-affine coordinates
of $\tau_{\text{WK}}(\bm t/2)$ and $\tau_{\text{BGW}}(\bm t/2)$,
and then we can apply Theorem \ref{thm-intro-formula} to compute the free energies.
The generating series of the BKP-affine coordinates
have simple expressions
(in terms of the first two basis vectors of the corresponding elements
in the Sato-Grassmannian):
\begin{equation*}
\wA^{\bullet}(w,z) = A^{\bullet}(w,z)- \frac{w-z}{4(w+z)}
= \frac{\Phi_1^\bullet(-z)\Phi_2^\bullet(-w)-
\Phi_1^\bullet(-w)\Phi_2^\bullet(-z)}{4(w+z)},
\end{equation*}
where $\bullet=\text{WK}$ or $\text{BGW}$.
The vectors $\Phi_1^{\text{WK}}(z),\Phi_2^{\text{WK}}(z)$ are the Faber-Zagier series:
\begin{equation*}
\Phi_1^{\text{WK}}(z)= \sum_{m=0}^\infty \frac{(6m-1)!!}{36^m\cdot (2m)!} z^{-3m},\quad
\Phi_2^{\text{WK}}(z)= -\sum_{m=0}^\infty \frac{(6m-1)!!}{36^m \cdot (2m)!} \frac{6m+1}{6m-1} z^{-3m+1};
\end{equation*}
and
$\Phi_1^{\text{BGW}}(z),\Phi_2^{\text{BGW}}(z)$ are:
\begin{equation*}
\begin{split}
&\Phi_1^{\text{BGW}} (z) =\sum_{k=0}^\infty \frac{\big((2k-1)!!\big)^2}{8^k\cdot k!} z^{-k},
\quad
\Phi_2^{\text{BGW}} (z) = z- \sum_{k= 0}^\infty \frac{(2k-1)!!(2k+3)!!}{8^{k+1}\cdot (k+1)!}z^{-k}.
\end{split}
\end{equation*}

The rest of this paper is arranged as follows.
In \S \ref{sec-pre} we recall some preliminaries of BKP tau-functions and the boson-fermion correspondence.
In \S \ref{sec-npt-disconn} we represent the fermionic and bosonic $n$-point functions
in terms of the affine coordinates.
In \S \ref{sec-conn} we compute the connected $n$-point functions
using results of \S \ref{sec-npt-disconn}.
In \S \ref{sec-kdv},
we prove the relation \eqref{eq-intro-KPBKP} for a KdV tau-function.
Finally in \S \ref{sec-WK-BGW},
we apply our methods to the Witten-Kontsevich tau-function and the BGW tau-function.

\section{Preliminaries of BKP Hierarchy and Boson-Fermion Correspondence}
\label{sec-pre}

In this section we
first give a brief review of the neutral fermions and boson-fermion correspondence
for the BKP hierarchy. See \cite{yo, djkm, jm} for details.
Then we recall the affine coordinates of a BKP tau-function, see \cite{hb}.

\subsection{Strict partitions and Schur Q-functions}

First we recall the definition of strict partitions and Schur Q-functions \cite{sch}.
See e.g. \cite{mac} for details.

A partition of an integer $n$ is a sequence of integers $\mu=(\mu_1,\cdots,\mu_l)$
such that $\mu_1\geq \cdots \geq\mu_l>0$ and $|\mu|:=\mu_1+\cdots+\mu_n = n$.
The number $l(\mu):=l$ is called the length of $\mu$.
A partition $\mu$ is called strict
if $\mu_1>\mu_2\cdots>\mu_l>0$.
The set of all strict partitions is denoted by $DP$,
and we allow the empty partition $(\emptyset) \in DP$ of length zero.
A partition $\mu$ is called odd
if each $\mu_i$ is odd.
The set of all odd partitions of $n$ is denoted by $OP_n$.
Let $\bm x = (x_1,x_2,\cdots)$ be a family of variables,
and let $p_n := \sum_i x_i^n$ be the Newton symmetric function of degree $n$.
Define:
\be
q_n = \sum_{\mu\in OP_n} \frac{2^{l(\mu)}}{\prod_{i\geq 1: \text{ odd}} i^{m_1} \cdot m_i!}
p_\mu,
\ee
where $p_\mu:= p_{\mu_1}\cdots p_{\mu_l}$ for a partition $\mu=(\mu_1,\cdots,\mu_l)$,
and $m_i$ is the number of $i$'s appearing in $\mu$.
Then the Schur Q-function $Q_\lambda$ indexed by a strict partition $\lambda\in DP$
is defined as follows:
\begin{equation*}
Q_{(m)} (\bm x) := q_m,\qquad
Q_{(m,n)} (\bm x) := q_m q_n -2q_{m+1}q_{n-1}+\cdots +(-1)^n 2q_{m+n},
\end{equation*}
and for $\lambda=(\lambda_1,\cdots,\lambda_n)\in DP$ with $n\geq 4$ even,
$Q_\lambda$ is defined by the Pfaffian:
\begin{equation*}
Q_\lambda =
\Pf \left[
\begin{array}{cccc}
0 & Q_{(\lambda_1,\lambda_2)} & \cdots & Q_{(\lambda_1,\lambda_n)} \\
-Q_{(\lambda_1,\lambda_2)} & 0 & \cdots & Q_{(\lambda_2,\lambda_n)} \\
\vdots & \vdots & \cdots & \vdots \\
-Q_{(\lambda_1,\lambda_n)} & -Q_{(\lambda_2,\lambda_n)} & \cdots & 0\\
\end{array}
\right];
\end{equation*}
and for $\lambda=(\lambda_1,\cdots,\lambda_n)\in DP$ with $n\geq 3$ odd,
$Q_\lambda$ is defined by:
\begin{equation*}
Q_\lambda:= q_{\lambda_1}Q_{(\lambda_2,\cdots,\lambda_n)}
- q_{\lambda_2}Q_{(\lambda_1,\lambda_3,\cdots,\lambda_n)} +\cdots
+ q_{\lambda_n}Q_{(\lambda_1,\cdots,\lambda_{n-1})}.
\end{equation*}
We will use the convention $Q_{(\emptyset)}:=1$.
By definition,
$Q_\lambda$ is a symmetric function in $\bm x = (x_1,x_2,\cdots)$ of degree $|\lambda|$
for every $\lambda\in DP$,
i.e., it is a vector in the bosonic Fock space $\Lambda = \bC[p_1,p_2,\cdots]$.
Moreover,
it lies in the subspace $\bC[p_1,p_3,p_5,\cdots]$.

\begin{Remark}
Schur Q-functions are related to the characters of
projective representations of the symmetric groups $S_n$,
see \cite{sch, hh}.
\end{Remark}

\subsection{Neutral fermions and fermionic Fock space}

Let $\{\phi_m\}_{m\in \bZ}$ be a family of operators
satisfying the following anti-commutation relations:
\be
\label{eq-anticomm}
[\phi_m, \phi_n]_+ := \phi_m\phi_n+\phi_n\phi_m = (-1)^m \delta_{m+n,0}.
\ee
In particular, one has
$\phi_0^2=\half$,
and $\phi_n^2=0$ for $n\not=0$.
These operators $\{\phi_m\}$ are called the neutral fermions.
The fermionic Fock space $\cF_B$ for the BKP hierarchy is the vector space (over $\bC$)
of all formal (infinite) summations
\begin{equation*}
\sum
c_{k_1,\cdots,k_n}
\phi_{k_1} \phi_{k_2} \cdots \phi_{k_n} |0\rangle,
\qquad
c_{k_1,\cdots,k_n} \in\bC,
\end{equation*}
over $n\geq 0$, $k_1>\cdots >k_n \geq 0$,
where $|0\rangle$ is a vector (called the vacuum) satisfying:
\be
\label{eq-B-anni}
\phi_i |0\rangle = 0,
\qquad \forall i <0.
\ee
The operators $\{\phi_n\}_{n\geq 0}$ are called the fermionic creators,
while $\{\phi_n\}_{n<0}$ are called the fermionic annihilators.
The Fock space $\cF_B$ can be decomposed as follows:
\begin{equation*}
\cF_B = \cF_B^0 \oplus \cF_B^1,
\end{equation*}
where $\cF_B^0$ and $\cF_B^1$ are the subspaces with
even and odd numbers of the generators $\{\phi_i\}_{i\geq 0}$ respectively.
The subspace $\cF_B^0$ has a basis $\{|\mu\rangle\}_{\mu\in DP}$ indexed by all strict partitions.
Let $\mu\in DP $ be  a strict partition $\mu= (\mu_1>\cdots >\mu_n > 0) $,
then:
\be
|\mu\rangle := \begin{cases}
\phi_{\mu_1}\phi_{\mu_2}\cdots \phi_{\mu_n}|0\rangle, &\text{ for $n$ even;}\\
\sqrt{2}\cdot \phi_{\mu_1}\phi_{\mu_2}\cdots \phi_{\mu_n} \phi_0 |0\rangle, &\text{ for $n$ odd}.
\end{cases}
\ee

Now we recall the dual Fock space $\cF_B^*$ and the pairing between $\cF_B$ and $\cF_B^*$.
Let $\cF_B^*$ be the vector space spanned by:
\begin{equation*}
\langle 0|
\phi_{k_n}  \cdots \phi_{k_2} \phi_{k_1},
\qquad
k_1 < k_2 < \cdots < k_n \leq 0,
\quad n\geq 0,
\end{equation*}
where $\langle 0|$ is a vector satisfying:
\be
\label{eq-B-anni-2}
\langle 0| \phi_i = 0,
\qquad \forall i >0.
\ee
Then there is a nondegenerate pairing $\cF_B^* \times \cF_B \to \bC$ determined by
\eqref{eq-B-anni}, \eqref{eq-B-anni-2},
the anti-commutation relation \eqref{eq-anticomm},
and the requirements
$\langle 0 | 0 \rangle =1$ and $\langle 0|\phi_0 | 0\rangle =0$.
One easily checks that for arbitrary $k_1>k_2>\cdots >k_n\geq 0$,
\be
\label{eq-basic-VEV}
\langle 0 | \phi_{-k_n}\cdots \phi_{-k_1}
\phi_{k_1}\cdots \phi_{k_n} |0\rangle = \begin{cases}
(-1)^{k_1+\cdots +k_n}, &\text{ if $k_n\not=0$;}\\
\half\cdot (-1)^{k_1+\cdots +k_{n-1}},&\text{ if $k_n=0$.}
\end{cases}
\ee
In general,
the vacuum expectation value of a product of neutral fermions
can be computed using Wick's Theorem:
\begin{equation*}
\langle 0| \phi_{i_1}\phi_{i_2}\cdots \phi_{i_{2n}}|0\rangle
=\sum_{\substack{(p_1,q_1,\cdots,p_n,q_n)\\p_k<q_k, \quad p_1<\cdots<p_n}}
\text{sgn}(p,q)\cdot \prod_{j=1}^n
\langle 0| \phi_{i_{p_j}}\phi_{i_{q_j}} |0\rangle,
\end{equation*}
where $(p_1,q_1,\cdots,p_n,q_n)$ is a permutation of $(1,2,\cdots,2n)$,
and $\text{sgn}(p,q)$ denotes its sign ($\text{sgn}=1$ for an even permutation, and $\text{sgn}=-1$ for an odd one).

The normal-ordered product $:\phi_i\phi_j:$ of two neutral fermions
is defined by:
\be
\label{eq-def-normal-order}
:\phi_i\phi_j: = \phi_i\phi_j -\langle 0 |\phi_i\phi_j |0\rangle.
\ee
Then by \eqref{eq-basic-VEV},
the anti-commutation relation \eqref{eq-anticomm} is equivalent to the following
operator product expansion (OPE):
\be
\phi(w)\phi(z)
= :\phi(w)\phi(z): +i_{w,z}\frac{w-z}{2(w+z)},
\ee
where $\phi(z)$ is the fermionic field:
\be
\label{eq-ferm-field}
\phi(z) = \sum_{i\in \bZ} \phi_i z^i,
\ee
and $i_{w,z}$ means formally expanding on $\{|w|>|z|\}$:
\begin{equation*}
i_{w,z}\frac{w-z}{2(w+z)}:= \half +
\sum_{j=1}^\infty (-1)^{j} w^{-j} z^j.
\end{equation*}

\subsection{Boson-fermion correspondence}

Given an odd integer $n\in 2\bZ+1$,
define the Hamiltonian $H_n$ by:
\be
\label{eq-def-boson}
H_n = \half \sum_{i\in \bZ} (-1)^{i+1} \phi_i\phi_{-i-n}.
\ee
Then $H_n |0\rangle = 0$
for $\forall n> 0$.
Moreover,
the following commutation relation holds:
\be
[H_n,H_m] = \frac{n}{2}\cdot \delta_{m+n,0},
\qquad
\forall n,m \text{ odd}.
\ee

Now let $\bm t=(t_1,t_3,t_5,t_7,\cdots)$ be a family of formal variables,
and define:
\be
H_+(\bm t) = \sum_{n>0: \text{ odd}} t_n H_n,
\ee
then the so-called boson-fermion correspondence is the following:
\begin{Theorem}
[\cite{djkm}]
\label{thm-bf}
There is an isomorphism of vector spaces:
\begin{equation*}
\sigma_B :\cF_B \to \bC[\![w;t_1,t_2,\cdots]\!]/\sim,
\qquad
|U\rangle \mapsto \sum_{i=0}^1 \omega^i\cdot \langle i |e^{H_+(\bm t)}|U\rangle,
\end{equation*}
where $\omega^2\sim 1$,
and $\langle 1| =\sqrt{2}\langle 0|\phi_0 \in (\cF_B^1)^*$.
Under this isomorphism,
one has:
\be
\label{eq-bfcor-boson}
\sigma_B (H_n |U\rangle) = \frac{\pd}{\pd t_n} \sigma_B(|U\rangle),
\qquad
\sigma_B (H_{-n}|U\rangle)=\frac{n}{2} t_n \cdot \sigma_B(|U\rangle),
\ee
for every odd $n>0$.
Moreover,
\be
\label{eq-bfcor-ferm}
\sigma_B (\phi(z)|U\rangle) =
\frac{1}{\sqrt{2}} \omega\cdot e^{\xi(\bm t,z)} e^{-\xi (\tilde\pd ,z^{-1})} \sigma_B(|U\rangle),
\ee
where
$\xi(\bm t,z) = \sum_{n>0\text{ odd}} t_{n}z^{n}$
and $\tilde\pd = (2\pd_{t_{ 1}},\frac{2}{3}\pd_{t_{ 3}},\frac{2}{5}\pd_{t_{ 5}},\cdots)$.
\end{Theorem}

Furthermore, one has the following:
\begin{Theorem}
[\cite{yo}]
\label{eq-bf-schur}
Let $\lambda = (\lambda_1>\cdots>\lambda_l>0)\in DP$, and take:
\begin{equation*}
t_n = \frac{2p_n}{n} = \frac{2}{n}\sum_i x_i^n,
\qquad \text{$n$ odd}
\end{equation*}
in $H_+(\bm t)$.
Then:
\begin{equation*}
Q_\lambda(\bm x) = 2^{\half l(\lambda)}\cdot \sigma_B (\phi_{\lambda_1}\cdots \phi_{\lambda_l}|\alpha(\lambda)\rangle),
\end{equation*}
where $Q_\lambda$ is the Schur Q-function indexed by $\lambda\in DP$,
and
\begin{equation*}
|\alpha(\lambda)\rangle = \begin{cases}
|0\rangle, & \text{if $l(\lambda)$ is even;}\\
\sqrt{2}\phi_0|0\rangle, & \text{if $l(\lambda)$ is odd}.
\end{cases}
\end{equation*}
\end{Theorem}

\subsection{BKP tau-functions and their affine coordinates}
\label{sec-pre-affine}

The BKP hierarchy is introduced by Kyoto School in \cite{djkm}.
A tau-function $\tau = \tau(\bm t)$ of the BKP hierarchy
is the image of a vector $e^g|0\rangle \in \cF_0$ under the boson-fermion correspondence:
\begin{equation*}
\tau(\bm t) = \langle 0| e^{H_+(\bm t)} e^g |0\rangle,
\end{equation*}
where $\bm t = (t_1,t_3,t_5,\cdots)$
and $g$ is of the form
$g=\sum_{m,n\in \bZ} c_{m,n}:\phi_m\phi_n:$
such that
\be
\label{eq-constraints-goinfty}
c_{m,n}=0,\qquad \text{ for $|m-n|>>0$}.
\ee

An alternative way to express a BKP tau-function
(or more precisely,
an element in the big cell of the isotropic Grassmannian)
in the fermionic Fock space
is to use Bogoliubov transforms which involves only fermionic creators.
Now we recall this approach
(see \cite[\S 7]{hb} for details).
Consider the following vector in $\cF_B$:
\be
\label{eq-Bog-tau}
|A\rangle := e^A |0\rangle \in \cF_B^0,
\ee
where
\be
\label{eq-bog-A}
A= \sum_{n,m\geq 0} a_{n,m} \phi_m \phi_n,
\qquad a_{n,m}\in \bC,
\ee
are quadratic in the fermionic creators.
(Here we do not impose constraints such like \eqref{eq-constraints-goinfty} on the coefficients.)
Recall that for $n,m\geq 0$,
one always has
$\phi_m \phi_n = -\phi_n\phi_m$
unless $n=m=0$,
thus we will always assume that:
\be
\label{eq-bog-A-cond}
a_{n,m} = -a_{m,n},
\qquad \forall n,m\geq 0.
\ee
In particular,
$a_{n,n}=0$ for every $n\geq 0$.
For an operator $A$ of this form,
the following identity will be useful
(which can  be proved using the Baker-Campbell-Hausdorff formula,
see e.g. \cite[\S 7.3.4]{hb} for details):
\be
\label{eq-lem-conj}
e^{-A} \phi_i e^A
= \begin{cases}
\phi_i, & \text{ $i>0$;}\\
\phi_i+\sum_{m\geq 0} 2(-1)^i
(-a_{-i,m}+a_{-i,0} a_{0,m})\phi_m, & \text{ $i\leq 0$}.
\end{cases}
\ee

One can expand the exponential $\exp(A)$ and represent the vector $e^A|0\rangle\in \cF_B^0$ in terms of the basis
$\{|\mu\rangle\}_{\mu\in DP}$ for $\cF_B^0$.
Given an arbitrary strict partition $\mu\in DP$,
we can always regard it as a partition of even length.
That means,
if $\mu=(\mu_1>\cdots>\mu_k > 0)$ where $k$ is odd,
then we will assign an additional summand $\mu_{k+1}:=0$ to $\mu$.
Denote by $\tilde l(\mu) \in 2\bZ$ this modified length of $\mu$:
\begin{equation*}
\tilde l ( \mu ) = \begin{cases}
l(\mu), & \text{ if $l(\mu)$ is even;}\\
l(\mu)+1, & \text{ if $l(\mu)$ is odd,}\\
\end{cases}
\end{equation*}
Then one has:
\be
\label{eq-bog-expand}
e^A |0\rangle = \sum_{i=0}^\infty \frac{1}{i!}
\big( \sum_{n,m\geq 0} a_{n,m} \phi_m \phi_n \big)^i
|0\rangle
= \sum_{\mu\in DP:\text{ $l(\mu)$ even}} c_\mu\cdot |\mu\rangle.
\ee
The coefficients $c_\mu$ are:
\be
\label{eq-Cartan-coord}
c_\mu = (-2)^{\tilde l(\mu)/2} \cdot \Pf(a_{\mu_i,\mu_j})_{1\leq i,j\leq \tilde l(\mu)},
\ee
where $\Pf(a_{\mu_i,\mu_j})$ is the Pfaffian of this anti-symmetric matrix
of size $\tilde l(\mu) \times \tilde l(\mu)$.
This is a straightforward consequence of Wick's Theorem.

Now one consider the image of a Bogoliubov transformation of the above form
under the boson-fermion correspondence.
Denote:
\be
\label{eq-def-tauA}
\tau_A := \sigma_B (e^A|0\rangle) =\langle0| e^{H_+(\bm t)}e^A|0\rangle \in \bC[\![p_1,p_3,p_5,\cdots]\!],
\ee
then by \eqref{eq-bog-expand} and Theorem \ref{eq-bf-schur} we have the following expansion by Schur Q-functions:
\be
\begin{split}
\tau_A =& \sum_{\mu\in DP} c_\mu \cdot 2^{-\tilde l(\mu)/2} Q_\mu(\bm x)\\
=& \sum_{\mu \in DP} (-1)^{\tilde l(\mu)/2} \cdot
\Pf (a_{\mu_i,\mu_j})_{1\leq i,j\leq \tilde l(\mu)} \cdot Q_\mu(\bm x),
\end{split}
\ee
where the time-variables are taken to be
$t_n = \frac{2p_n}{n} =\frac{2}{n}\sum x_i^n$ for every odd $n>0$ in the boson-fermion correspondence.
The first a few terms of $\tau_A$ are:
\begin{equation*}
\begin{split}
\tau_A =& 1+ \sum_{n>0}a_{0,n}\cdot Q_{(n)}(\bm x)
+\sum_{m>n>0} a_{n,m}\cdot Q_{(m,n)}(\bm x)\\
& + \sum_{m>n>l>0} (a_{n,m}a_{0,l} -a_{l,m}a_{0,n}+ a_{0,m}a_{l,n})Q_{(m,n,l)} (\bm x)\\
& +\sum_{m>n>l>k>0} (a_{n,m}a_{k,l} -a_{l,m}a_{k,n}+ a_{k,m}a_{l,n})Q_{(m,n,l,k)} (\bm x) +\cdots.
\end{split}
\end{equation*}

The function $\tau_A$ is a tau-function of the BKP hierarchy.
Moreover, a formal power series tau-function $\tau=\tau(\bm t)$
with constant term $\tau(0)=1$ can be uniquely represented as $\tau=\tau_A$ for an operator $A$
of the form \eqref{eq-bog-A} satisfying the anti-symmetric condition \eqref{eq-bog-A-cond}
(see \cite[Theorem 7.3.7]{hb}).
In fact,
the coefficient $a_{n,m}$ is exactly the coefficient of $Q_{(m,n)}$
in the Schur Q-function expansion of $\tau$
for $m> n\geq 0$.
The coefficients $\{a_{n,m}\}_{n,m\geq 0}$ are called the affine coordinates of this tau-function.
In the rest of this paper,
we will always assume $\tau(0) =1$.

\subsection{Isotropic Sato Grassmannian}

The affine coordinates $\{a_{n,m}\}$ discussed above provide a natural choice of coordinates on the big-cell of
the isotropic Sato Grassmannian associated to the BKP hierarchy.
Here we end this section by giving a brief review of the isotropic Grassmannian.
See e.g. \cite[\S 7]{hb} for an introduction.

Let $\cH =  \cH_+ \oplus \cH_-$,
where $\cH_+ := \text{span}\{z^i\}_{i\geq 0}$ and $\cH_- := \text{span}\{z^i\}_{i< 0}$,
then $\{e_j:=z^{-j-1}\}_{j\in \bZ}$ form a basis for $\cH$.
Let $\{\tilde e^i\}_{i\in \bZ}$ be the dual basis for $\cH^*$,
and define $\cH_\phi \subset \cH\oplus \cH^*$ to be the linear subspace
spanned by $\{e_i^0\}_{i\in \bZ}$,
where:
\begin{equation*}
e_i^0 := \frac{1}{\sqrt{2}} (e_i + (-1)^i \tilde e^{-i}).
\end{equation*}
Let $Q_\phi :\cH_\phi\times \cH_\phi \to \bC$ be the nondegenerate symmetric bilinear form on $\cH_\phi$ satisfying
$Q_\phi(e_i^0,e_j^0) = (-1)^{i+j} \delta_{i,-j}$.
Then the subspace
$\cH_\phi^0 := \text{span} \{e_i^0\}_{i<0} \subset \cH_\phi$
is maximally totally isotropic with respect to $Q_\phi$.
The Grassmannian $Gr_{\cH_\phi^0}^0 (\cH_\phi)$ is defined to be the orbit of $\cH_\phi^0$
under the action of the orthogonal group:
\begin{equation*}
O(\cH_\phi,Q_\phi) := \{g\in GL(\cH_\phi) \big| Q_\phi(gu,gv) =Q_\phi(u,v), \forall u,v \in \cH_\phi\}.
\end{equation*}

Let $w^0 \in Gr_{\cH_\phi^0}^0 (\cH_\phi)$ be a maximally isotropic subspace of $\cH_\phi$,
and assume that $w^0 = \text{span} \{w_i\}_{i>0}$.
Then one can associate an element $Ca(w^0)$
in the projectivization $\bP (\cF_B)$ of the fermionic Fock space to $w^0$ by defining
$Ca(w^0) := [\prod_{i>0} \phi_{w_i} |0\rangle]$,
where for $w_i = \sum_{j\in \bZ} c_i^j e_j^0 \in \cH_\phi$
we denote $\phi_{w_i} = \sum_{j\in \bZ} c_i^j \phi_j$.
This defines a map
\begin{equation*}
Ca:
Gr_{\cH_\phi^0}^0 (\cH_\phi)\to \bP (\cF_B),\qquad
w^0 \mapsto Ca(w^0),
\end{equation*}
which is called the Cartan map.
It is the infinite-dimensional version of the map introduced by Cartan \cite{ca}.
The image of an element in this Grassmannian under the Cartan map is of the form
(up to projectivization):
\be
\label{eq-image-cartan}
\sum_{\mu \in DP} \kappa_\mu \cdot |\mu\rangle,
\ee
where the coefficients $\{\kappa_\mu\}$
satisfy the so-called Cartan relations
(which are the analogue of the Pl\"ucker relations on the ordinary Grassmannian).
Moreover,
the Cartan relations are equivalent to the BKP Hirota bilinear relations,
thus this fermionic vector becomes a BKP tau-function under the boson-fermion correspondence.
The coefficients $\{\kappa_\mu\}$ are called the Cartan coordinates of this tau-function.
When $w^0$ lies in a certain subspace (called the big-cell) of the isotropic Grassmannian,
the fermionic vector corresponding to this BKP tau-function can be uniquely represented as
a Bogoliubov transform of the form \eqref{eq-Bog-tau},
and its Cartan coordinates $\{\kappa_\mu\}$ are given by the Pfaffians $\Pf(a_{\mu_i,\mu_j})_{1\leq i,j\leq \tilde l(\mu)}$
of the affine coordinates $\{a_{n,m}\}$.
In particular,
$a_{n,m}$ is exactly the Cartan coordinate indexed by the strict partition $\mu=(m>n\geq 0)$.
For details,
see the book \cite[Theorem 7.1.1; \S 7.3]{hb}.

\section{Computation of Fermionic and Bosonic $N$-Point Functions}
\label{sec-npt-disconn}

In this section,
we compute the fermionic and bosonic $n$-point functions associated to a tau-function of the BKP hierarchy.
We represent the results in terms of the generating series of the affine coordinates $\{a_{n,m}\}_{n,m\geq 0}$.

\subsection{Bosonic and fermionic $n$-point functions}

Let $n\geq 1$ be a positive integer,
and let $A$ be an operator of the form \eqref{eq-bog-A} satisfying the condition \eqref{eq-bog-A-cond}.
Similar to the case of KP tau-functions (see \cite[\S 4]{zhou1}),
here we consider the bosonic $n$-point functions associated to
a BKP tau-function $\tau_A$ of the form \eqref{eq-def-tauA}:
\be
\langle H(z_1)\cdots H(z_n) \rangle_A := \langle 0| H(z_1)\cdots H(z_n) e^A |0\rangle,
\ee
where $z_1,\cdots,z_n$ are some formal variables, and
\be
H(z)= \sum_{n\in \bZ: \text{ odd}} H_n z^{-n}
\ee
is the generating series of the bosons $H_n$ defined by \eqref{eq-def-boson}.

Our goal in this section is to represent the bosonic $n$-point functions in terms of
the affine coordinates $\{a_{n,m}\}_{n,m\geq 0}$.
In order to do that,
we will need to compute the following fermionic $n$-point functions first:
\be
\langle \phi(z_1)\cdots \phi(z_n) \rangle_A := \langle 0| \phi(z_1)\cdots \phi(z_n) e^A |0\rangle,
\ee
where $\phi(z)$ is the generating series \eqref{eq-ferm-field} of neutral fermions.

\subsection{Fermionic $2$-point function in terms of affine coordinates}

In this subsection,
we derive a formula for the fermionic $2$-point function $\langle \phi(w)\phi(z) \rangle_A$
in terms of the generating series of the affine coordinates.

First denote:
\begin{equation*}
\phi(w)_+ = \sum_{i> 0} w^i \phi_i,
\qquad\qquad
\phi(w)_- = \sum_{i<0} w^i \phi_i,
\end{equation*}
then one has $\phi(w)=\phi(w)_+ + \phi_0 + \phi(w)_-$,
and:
\be
\label{eq-ferm-2pt-decomp}
\begin{split}
\phi(w)\phi(z) =&
\phi(w)_+ \phi(z)_+ +
\phi(w)_+ \phi(z)_- +
\phi(w)_- \phi(z)_+ +
\phi(w)_- \phi(z)_-\\
&+ \phi(w)_+\phi_0 + \phi(w)_-\phi_0
+\phi_0\phi(z)_+ + \phi_0\phi(z)_- +\half.
\end{split}
\ee
Now we compute the right-hand side of \eqref{eq-ferm-2pt-decomp} term by term.
Since $\phi_-(z)|0\rangle=0$,
and $\langle 0 | \phi_{i_1}\cdots \phi_{i_s} |0\rangle = 0$ unless
$s$ is even and $i_1,\cdots,i_s$ contains an equal number of positive and negative integers,
we easily see:
\begin{equation*}
\begin{split}
&\langle \phi_+(w) \phi_+(z) \rangle_A =0,\qquad
\langle \phi_+(w) \phi_-(z) \rangle_A =0,\\
&\langle \phi(w)_+ \phi_0 \rangle_A
= \langle \phi_0 \phi(z)_+ \rangle_A =0,
\end{split}
\end{equation*}
and
\begin{equation*}
\langle \phi_-(w) \phi_+(z) \rangle_A
= \langle \sum_{j<0, i> 0} w^jz^i \phi_j\phi_i \rangle_A
= -i_{w,z} \frac{z}{z+w},
\end{equation*}
where
\be
i_{w,z} \frac{z}{z+w} := \sum_{i=1}^\infty (-1)^{i+1} w^{-i} z^i.
\ee
Moreover, we have:
\begin{equation*}
\begin{split}
\langle \phi_-(w) \phi_-(z) \rangle_A
=& \langle \sum_{i,j<0} w^j z^i \phi_j\phi_i \cdot
\Big(\sum_{m,n> 0} a_{n,m} \phi_m\phi_n
\Big)\rangle \\
=& \sum_{n,m>0} (-1)^{n+m}a_{n,m} (w^{-n}z^{-m}-w^{-m}z^{-n})\\
=& \sum_{n,m>0} (-1)^{n+m} \cdot 2a_{n,m} w^{-n} z^{-m},
\end{split}
\end{equation*}
where in the last step we have used the anti-symmetry property $a_{n,m}=-a_{m,n}$.
\begin{Remark}
In the expansion of $e^A$
there are terms of the form $a_{n,0}a_{m,0}\phi_0\phi_n\phi_0\phi_m$.
However,
the total contribution of all such terms to $\langle \phi_-(w) \phi_-(z) \rangle_A$
turns out to be zero  due to the anti-commutation relation of $\{\phi_i\}_{i>0}$
and the anti-symmetry property $a_{n,m}=-a_{m,n}$.
In fact, one has:
\begin{equation*}
\begin{split}
& \sum_{n,m> 0}
a_{n,0}a_{m,0}\phi_0\phi_n\phi_0\phi_m
= \sum_{n,m> 0}
(-\half a_{n,0}a_{m,0}\cdot
 \phi_n\phi_m ),\\
&\sum_{n,m> 0}
a_{0,n}a_{m,0}\phi_n\phi_0\phi_0\phi_m
=-\sum_{n,m> 0} \half a_{n,0}a_{m,0}\phi_n\phi_m
= -\sum_{n,m> 0} \half a_{m,0}a_{n,0}\phi_m\phi_n,
\end{split}
\end{equation*}
where in the second step of the second equality we have exchanged the indices $m,n$.
Therefore the total contribution of
\begin{equation*}
\sum_{n,m> 0}
a_{n,0}a_{m,0}\phi_0\phi_n\phi_0\phi_m
+
\sum_{n,m> 0}
a_{0,n}a_{m,0}\phi_n\phi_0\phi_0\phi_m
\end{equation*}
to $\langle \phi_-(w) \phi_-(z) \rangle_A$ is:
\begin{equation*}
-\half \langle \sum_{i,j<0} w^j z^i \phi_j\phi_i \cdot
\sum_{n,m> 0} a_{n,0} a_{m,0}(\phi_n\phi_m + \phi_m\phi_n)
\rangle =0.
\end{equation*}
Similarly,
the total contribution of
\begin{equation*}
\sum_{n,m>0}a_{n,0}a_{0,m}\phi_0\phi_n\phi_m\phi_0
+\sum_{n,m>0} a_{0,n}a_{0,m}\phi_n\phi_0\phi_m\phi_0
\end{equation*}
is also zero,
thus $\langle \phi_-(w) \phi_-(z) \rangle_A$ does not contain terms of the form
$a_{n,0}a_{m,0}$.
\end{Remark}

Finally, we have:
\begin{equation*}
\begin{split}
&\langle \phi(w)_- \phi_0 \rangle_A
=\langle
\sum_{i<0} w^i\phi_i\phi_0 \cdot \sum_{n>0} \big(
a_{n,0}\phi_0\phi_n +a_{0,n} \phi_n\phi_0\big) \rangle_A
= \sum_{n>0}(-1)^n a_{n,0} w^{-n},\\
&\langle \phi_0 \phi(z)_- \rangle_A
=\langle
\sum_{i<0} z^i\phi_0\phi_i \cdot \sum_{n>0} \big(a_{n,0}\phi_0\phi_n
+a_{0,n}\phi_n\phi_0 \big) \rangle_A
= -\sum_{n>0}(-1)^n a_{n,0} z^{-n},
\end{split}
\end{equation*}
since $\phi_0^2 = \half$.
Thus by \eqref{eq-ferm-2pt-decomp} we conclude that:
\begin{Proposition}
\label{prop-ferm-2pt}
The fermionic $2$-point function is given by:
\be
\langle \phi(w)\phi(z) \rangle_A = -2A(w,z)
+ i_{w,z}\frac{w-z}{2(w+z)},
\ee
where $A(w,z)$ is the following generating series of $\{a_{n,m}\}$:
\be
\label{eq-def-seriesA}
A(w,z) = \sum_{n,m>0} (-1)^{m+n+1} \cdot a_{n,m} w^{-n} z^{-m}
-\half \sum_{n>0}(-1)^n a_{n,0} (w^{-n}-z^{-n}),
\ee
and
\be
i_{w,z}\frac{w-z}{2(w+z)} = \half -
i_{w,z} \frac{z}{z+w} = \half + \sum_{i=1}^\infty (-1)^{i} w^{-i} z^i.
\ee
\end{Proposition}

In what follows,
we will also use the following notation:
\be
\label{eq-def-Ahat}
\wA(w,z) := -\half \langle \phi(w)\phi(z) \rangle_A=  A(w,z)-i_{w,z}\frac{w-z}{4(w+z)}.
\ee

\begin{Remark}
The assumption $a_{n,m}=-a_{m,n}$ implies that $A(w,z)$ is anti-symmetric:
\be
A(w,z) = -A(z,w).
\ee
\end{Remark}

\subsection{Fermionic $n$-point functions for general $n$}

In this subsection we compute the fermionic $n$-point functions
$\langle \phi(z_1) \phi(z_2)\cdots \phi(z_n) \rangle_A$ for general $n$.

It is clear that $\langle \phi(z_1) \cdots \phi(z_n) \rangle_A=0$ if $n$ is odd.
And if $n=2s$ is even,
\begin{equation*}
\begin{split}
\langle \phi(z_1) \cdots \phi(z_{2s}) \rangle_A
=& \sum_{i_1,\cdots,i_n}
z_1^{i_1}\cdots z_{2s}^{i_{2s}}\langle 0| \phi_{i_1}\phi_{i_2}\cdots\phi_{i_{2s}} e^A|0\rangle\\
=& \sum_{i_1,\cdots,i_n}
z_1^{i_1}\cdots z_{2s}^{i_{2s}}\langle 0| (e^{-A}\phi_{i_1}e^A)(e^{-A}\phi_{i_2}e^A)\cdots (e^{-A}\phi_{i_{2s}} e^A)|0\rangle
\end{split}
\end{equation*}
since $\langle 0|e^{-A} = \langle 0|$.
By \eqref{eq-lem-conj} we know that $e^{-A}\phi_k e^A$ is a linear combination
of the neutral fermions $\{\phi_i\}_{i\in \bZ}$,
thus we can apply Wick's theorem and get:
\begin{equation*}
\begin{split}
&\langle 0| (e^{-A}\phi_{i_1}e^A)(e^{-A}\phi_{i_2}e^A)\cdots (e^{-A}\phi_{i_{2s}} e^A)|0\rangle\\
=&\sum_{\substack{(p_1,q_1,\cdots,p_s,q_s)\\p_k<q_k, \quad p_1<\cdots<p_s}}
\text{sgn}(p,q)\cdot \prod_{j=1}^s \langle 0| (e^{-A}\phi_{i_{p_j}}e^A)(e^{-A}\phi_{i_{q_j}}e^A) |0\rangle\\
=&\sum_{\substack{(p_1,q_1,\cdots,p_s,q_s)\\p_k<q_k, \quad p_1<\cdots<p_s}}
\text{sgn}(p,q)\cdot \prod_{j=1}^s \langle \phi_{i_{p_j}}\phi_{i_{q_j}}\rangle_A,
\end{split}
\end{equation*}
where $(p_1,q_1,\cdots,p_s,q_s)$ is a permutation of $(1,2,\cdots,2s)$ and
$\text{sgn}(p,q)$ is the sign of this permutation.
Therefore,
\be
\label{eq-Wick-ferm}
\langle \phi(z_1) \cdots \phi(z_{2s}) \rangle_A
=\sum_{\substack{(p_1,q_1,\cdots,p_s,q_s)\\p_k<q_k, \quad p_1<\cdots<p_s}}
\text{sgn}(p,q)\cdot \prod_{j=1}^s \langle \phi(z_{p_j})\phi(z_{q_j})\rangle_A.
\ee
This is equivalent to say that $\langle \phi(z_1) \cdots \phi(z_{2s}) \rangle_A$ equals to
the Pfaffian of the anti-symmetric matrix of size $2s \times 2s$
whose upper-triangular part is given by the fermionic two-point functions
$(\langle\phi(z_i)\phi(z_j)\rangle_A)$ for $1\leq i<j\leq 2s$.
Then by Proposition \ref{prop-ferm-2pt},
we conclude that:
\begin{Proposition}
\label{prop-ferm-npt}
We have:
\be
\langle \phi(z_1)\cdots \phi(z_{n}) \rangle_A = \begin{cases}
0, & \text{ if $n$ is odd;}\\
\Pf \big(\wB_{ij}\big)_{1\leq i,j\leq n}, & \text{ if $n$ is even},
\end{cases}
\ee
where
\be
\wB_{ij} =
\begin{cases}
-2\wA (z_i,z_j), & \text{ if $i<j$;}\\
0, & \text{ if $i=j$;}\\
2\wA (z_j,z_i),  & \text{ if $i>j$,}\\
\end{cases}
\ee
and $\wA(w,z)$ is the generating series \eqref{eq-def-Ahat} of the affine coordinates.
\end{Proposition}

\begin{Remark}
Notice here we cannot take $B_{ij} = -2\wA(z_i,z_j)$ directly for $i>j$,
since in general $\wA(z_i,z_j) \not = -\wA(z_j,z_i)$.
In fact,
\begin{equation*}
\wA(w,z)+\wA(z,w) = -\half \sum_{n\in \bZ} \Big(\frac{-z}{w}\Big)^n = -\half \delta(-z,w),
\end{equation*}
where $\delta$ is the formal delta-function.
\end{Remark}

\subsection{Representing bosonic fields in terms of fermionic fields}

In this subsection,
we first recall the fact that the bosonic field $H(z)$ is the normal-ordered product of
two fermionic fields.
This relation will provide us a way to compute the bosonic $n$-point functions
using the above results about the fermionic $n$-point functions.
We deal with the simplest case $n=1$ in this subsection,
and compute $\langle H(z_1)\cdots H(z_n)\rangle_A$ for general $n$ in the next subsection.

By the definition \eqref{eq-def-boson} we see:
\be
\label{eq-ham-gen-1}
\begin{split}
H(z) =& \half \sum_{n\in \bZ:\text{ odd}} z^{-n}\Big(
\sum_{i\in \bZ} (-1)^{i+1} \phi_i\phi_{-i-n}\Big)\\
=& -\half
\sum_{n\in \bZ:\text{ odd}} \sum_{i\in \bZ}
(-z)^i\phi_i \cdot z^{-i-n}\phi_{-i-n}.
\end{split}
\ee
Notice that one can also define $H_{2k}$ for a nonzero integer $k$ using \eqref{eq-def-boson},
and the anti-commutation relation \eqref{eq-anticomm} implies:
\begin{equation*}
H_{2k}=0,\qquad \forall k\not= 0,
\end{equation*}
immediately.
Thus the equality \eqref{eq-ham-gen-1} can be rewritten as:
\be
\label{eq-ham-gen-2}
\begin{split}
H(z) =& -\half
\sum_{n\not= 0} \sum_{i\in \bZ}
(-z)^i\phi_i \cdot z^{-i-n}\phi_{-i-n}\\
=& -\half \sum_{n\not= 0} \sum_{i\in \bZ}
(-z)^i\cdot z^{-i-n} :\phi_i\phi_{-i-n}:,
\end{split}
\ee
since $:\phi_i\phi_j: =\phi_i\phi_j$ if $i+j\not=0$.
Moreover,
recall $:\phi_0^2: = \phi_0^2 -\langle 0|\phi_0^2|0\rangle=0$,
and
\begin{equation*}
\sum_{i\not=0} (-z)^{-i} z^i :\phi_{-i}\phi_i:
=\sum_{i>0} (-1)^i (\phi_{i}\phi_{-i} -\phi_{i}\phi_{-i})=0,
\end{equation*}
thus one can rewrite \eqref{eq-ham-gen-2} as follows:
\be
H(z) =-\half
\sum_{n\in \bZ} \sum_{i\in \bZ}
(-z)^i \cdot z^{-i-n} :\phi_i\phi_{-i-n}:,
\ee
or equivalently,
\be
\label{eq-H-normal}
H(z) =-\half
:\phi(-z)\phi(z):.
\ee

Now we are able to compute the bosonic $n$-point functions using the relation \eqref{eq-H-normal} and
the results for the fermionic $n$-point functions.
The normal-ordered product of two fermionic fields are given by the OPE:
\be
\label{eq-OPE-ferm}
:\phi(w)\phi(z) : =\phi(w) \phi(z) -i_{w,z}\frac{w-z}{2(w+z)},
\ee
thus:
\begin{equation*}
\begin{split}
\langle :\phi(w)\phi(z): \rangle_A
=& \langle \phi(w)\phi(z) \rangle_A -
\langle 0|  i_{w,z}\frac{w-z}{2(w+z)}\cdot e^A|0\rangle\\
=&\langle \phi(w)\phi(z) \rangle_A - i_{w,z}\frac{w-z}{2(w+z)}.
\end{split}
\end{equation*}
Then by Proposition \ref{prop-ferm-2pt} we obtain:
\begin{Lemma}
We have
$\langle :\phi(w)\phi(z): \rangle_A = -2A(w,z)$.
\end{Lemma}

Furthermore,
let $w\to -z$ and use \eqref{eq-H-normal}, then we finally get:
\begin{Proposition}
The bosonic $1$-point function is given by:
\be
\langle H(z) \rangle_A = A(-z,z),
\ee
where $A(w,z)$ is the series \eqref{eq-def-seriesA}.
Or more explicitly,
\be
\langle H(z) \rangle_A
= \sum_{n,m>0} (-1)^{m+1} \cdot a_{n,m} z^{-(m+n)}
+\sum_{n>0}\epsilon_n\cdot (-1)^{n} a_{n,0} z^{-n},
\ee
where
\begin{equation*}
\epsilon_n = \begin{cases}
0,&\text{ $n$ even;}\\
1,&\text{ $n$ odd.}
\end{cases}
\end{equation*}
\end{Proposition}

\subsection{Bosonic $n$-point functions for general $n$}
\label{sec-bos-npt}

Now we can compute the bosonic $n$-point functions $\langle H(z_1)\cdots H(z_n) \rangle_A$ for general $n$.
First we prove the following:
\begin{Proposition}
We have:
\be
\langle :\phi(z_1)\phi(z_2)::\phi(z_3)\phi(z_4): \cdots :\phi(z_{2s-1})\phi(z_{2s}): \rangle_A
= \Pf(B_{ij})_{1\leq i,j \leq 2s},
\ee
where the entries $B_{ij}$ are defined as follows.
For $1\leq i < j \leq 2s$,
\be
B_{ij} = \begin{cases}
-2 A(z_i,z_j), & \text{ if $i=2r-1,j=2r$ for some $1\leq r\leq s$;}\\
-2\wA(z_i,z_j), & \text{ otherwise,}
\end{cases}
\ee
and $B_{ij} = -B_{ji}$ if $i>j$,
where $A(w,z)$ and $\wA(w,z)$ are given by \eqref{eq-def-seriesA} and \eqref{eq-def-Ahat} respectively.
And $B_{ii}=0$ for every $i$.
\end{Proposition}
\begin{proof}
First recall the OPE \eqref{eq-OPE-ferm}.
We have:
\be
\label{eq-pf-nptboson-1}
\begin{split}
&\langle :\phi(z_1)\phi(z_2)::\phi(z_3)\phi(z_4): \cdots :\phi(z_{2s-1})\phi(z_{2s}): \rangle_A\\
=& \langle
\Big( \phi(z_1)\phi(z_2) -i_{z_1,z_2} \frac{z_1-z_2}{2(z_1+z_2)} \Big) \cdots\\
&\quad \Big( \phi(z_{2s-1})\phi(z_{2s}) -  i_{z_{2s-1},z_{2s}} \frac{z_{2s-1}-z_{2s}}{2(z_{2s-1}+z_{2s})} \Big) \rangle_A\\
=& \sum_{K\sqcup L=\{1,2,\cdots,s\}}
\Big(\prod_{l\in L} f_l \Big)
\cdot \langle \phi_K \rangle_A,
\end{split}
\ee
where for a subset $K=\{k_1,\cdots,k_r\}\subset \{1,\cdots, s\}$ with $k_1<k_2<\cdots<k_r$,
denote:
\be
\phi_K = \phi(z_{2k_1-1})\phi(z_{2k_1})
\phi(z_{2k_2-1})\phi(z_{2k_2})\cdots
\phi(z_{2k_r-1})\phi(z_{2k_r}),
\ee
and
\be
f_l =
-i_{z_{2l-1},z_{2l}} \frac{z_{2l-1}-z_{2l}}{2(z_{2l-1}+z_{2l})}.
\ee
Then applying Wick's theorem (see \eqref{eq-Wick-ferm}) to $\langle\phi_K\rangle_A$ in \eqref{eq-pf-nptboson-1},
we get:
\be
\label{eq-pf-nptboson-2}
\begin{split}
&\langle :\phi(z_1)\phi(z_2)::\phi(z_3)\phi(z_4): \cdots :\phi(z_{2s-1})\phi(z_{2s}): \rangle_A\\
=& \sum_{K\sqcup L=\{1,2,\cdots,s\}}
\sum_{\substack{(p_1,q_1,\cdots,p_r,q_r)\\p_i<q_i,\quad p_1<\cdots<p_r}}
\text{sgn}(p,q)\cdot \Big(\prod_{l\in L} f_l \Big)\cdot \prod_{i=1}^r
\langle \phi(z_{p_i})\phi(z_{q_i}) \rangle_A,
\end{split}
\ee
where $(p_1,q_1,\cdots,p_r,q_r)$ is a permutation of $(2k_1-1,2k_1,2k_2-1,2k_2,\cdots,2k_r-1,2k_r)$
and $\text{sgn}(p,q)$ is the sign of this permutation.

On the other hand,
since for every $1\leq i<j\leq 2s$ one has:
\begin{equation*}
B_{ij} = \begin{cases}
\langle \phi(z_i)\phi(z_j) \rangle_A + f_t
,
& \text{ if $i=2t-1,j=2t$ for some $1\leq t\leq s$};\\
\langle \phi(z_i)\phi(z_j) \rangle_A, & \text{ otherwise},
\end{cases}
\end{equation*}
thus the Pfaffian of the matrix $B=(B_{ij})_{1\leq i,j \leq 2s}$ is:
\be
\label{eq-pf-nptboson-3}
\Pf (B) =
\sum_{\substack{(p_1',q_1',\cdots,p_s',q_s')\\ p_i'<q_i', \quad p_1'<\cdots<p_s'}}
\text{sgn}(p',q') \prod_{i=1}^s \Big(\langle \phi(z_{p_i'})\phi(z_{q_i'}) \rangle_A + f(p_i',q_i')\Big),
\ee
where $(p_1',q_1',\cdots,p_s',q_s')$ is a permutation of $(1,2,\cdots,2s)$,
and
\begin{equation*}
f(p_i',q_i') := \begin{cases}
f_t, & \text{ if $p_i'=2t-1,q_i'=2t$ for some $1\leq t\leq s$;}\\
0, & \text{ otherwise}.
\end{cases}
\end{equation*}
It is clear that \eqref{eq-pf-nptboson-2} and \eqref{eq-pf-nptboson-3} are equal,
thus we have proved the conclusion.
\end{proof}

Now recall that $H(w) = -\half  :\phi(-w)\phi(w):$.
Therefore the following theorem is proved by taking $z_{2i-1}\to -w_i$ and $z_{2i}\to w_i$.
\begin{Theorem}
\label{thm-bos-npt}
We have:
\be
\langle H(w_1)\cdots H(w_n) \rangle_A
= \Pf(C_{ij})_{1\leq i,j \leq 2n},
\ee
where the entries $C_{ij}$ are defined as follows.
For $1\leq i < j \leq 2n$,
\be
\label{eq-def-Cij}
C_{ij} = \begin{cases}
A(-w_r,w_r), & \text{ $i=2r-1,j=2r$ for some $1\leq r\leq n$};\\
\wA\big((-1)^i w_{\lceil \frac{i}{2} \rceil }, (-1)^j w_{\lceil \frac{j}{2} \rceil }), & \text{ otherwise},
\end{cases}
\ee
and $C_{ij} = -C_{ji}$ if $i>j$;
$C_{ii}=0$ for every $i$.
Here $\lceil \frac{i}{2} \rceil = \frac{i}{2}$ if $i$ is even, and $\lceil \frac{i}{2} \rceil = \frac{i+1}{2}$ if $i$ is odd.
\end{Theorem}

\begin{Example}
Let us present some examples.
For $n=1$,
we see:
\begin{equation*}
\langle H(w) \rangle_A = \Pf
\left[
\begin{array}{cc}
0 & A(-w,w) \\
-A(-w,w) & 0 \\
\end{array}
\right]
=A(-w,w).
\end{equation*}
For $n=2$, we have:
\begin{equation*}
\begin{split}
&\langle H(w_1) H(w_2) \rangle_U\\
=&
\Pf \left[
\begin{array}{cccc}
0 & A(-w_1,w_1) & \widehat A(-w_1,-w_2) &\widehat A(-w_1,w_2)\\
-A(-w_1,w_1) & 0 & \widehat A(w_1,-w_2) &\widehat A(w_1,w_2)\\
-\widehat A(-w_1,-w_2) & -\widehat A(w_1,-w_2) & 0 & A(-w_2,w_2) \\
-\widehat A(-w_1,w_2) & -\widehat A(w_1,w_2) & -A(-w_2,w_2) & 0\\
\end{array}
\right]\\
=& A(-w_1,w_1) A(-w_2,w_2)- \widehat A(-w_1,-w_2) \widehat A(w_1,w_2)
+ \widehat A(-w_1,w_2) \widehat A(w_1,-w_2).
\end{split}
\end{equation*}
And for $n=3$,
one can check that:
\begin{equation*}
\begin{split}
&\langle H(w_1)H(w_2)H(w_3) \rangle_A =
\wA(-w_1,w_3) \wA(w_1,-w_3) A(-w_2,w_2) \\
&\quad - \wA(-w_1,-w_3) \wA(w_1,w_3) A(-w_2,w_2)
- \wA(-w_1,w_3) \wA(w_1,w_2) \wA(-w_2,-w_3) \\
&\quad + \wA(-w_1,w_2) \wA(w_1,w_3) \wA(-w_2,-w_3)
+ \wA(-w_1,-w_3) \wA(w_1,w_2) \wA(-w_2,w_3) \\
&\quad - \wA(-w_1,w_2) \wA(w_1,-w_3) \wA(-w_2,w_3)
+ \wA(-w_1,w_3) \wA(w_1,-w_2) \wA(w_2,-w_3) \\
&\quad - \wA(-w_1,-w_2) \wA(w_1,w_3) \wA(w_2,-w_3)
+ A(-w_1,w_1) \wA(-w_2,w_3) \wA(w_2,-w_3) \\
&\quad - \wA(-w_1,-w_3) \wA(w_1,-w_2) \wA(w_2,w_3)
+ \wA(-w_1,-w_2) \wA(w_1,-w_3) \wA(w_2,w_3) \\
&\quad - A(-w_1,w_1) \wA(-w_2,-w_3) \wA(w_2,w_3)
+ \wA(-w_1,w_2) \wA(w_1,-w_2) A(-w_3,w_3) \\
&\quad - \wA(-w_1,-w_2) \wA(w_1,w_2) A(-w_3,w_3)
+ A(-w_1,w_1) A(-w_2,w_2) A(-w_3,w_3).
\end{split}
\end{equation*}

\end{Example}

\subsection{$A(w,z)$ as a specialization of tau-function}

Similar to the case of the KP hierarchy
(see \cite[\S 5.6]{zhou1}),
the generating series $A(w,z)$ defined by \eqref{eq-def-seriesA}
can be represented as a special evaluation of the tau-function $\tau_A(\bm t)$.
We show this in the present subsection.
This relation will be useful in \S \ref{sec-kdv}.

First recall the relation \eqref{eq-bfcor-ferm}.
One has:
\begin{equation*}
\begin{split}
\langle 0| e^{H_+(\bm t)} \phi(w)\phi(z) e^A|0\rangle
=& \half
e^{\xi(\bm t,w)} e^{-\xi (\tilde\pd ,w^{-1})}
e^{\xi(\bm t,z)} e^{-\xi (\tilde\pd ,z^{-1})} \tau_A(\bm t)\\
=& \half
e^{-[\xi (\tilde\pd ,w^{-1}),\xi(\bm t,z)]}
e^{\xi(\bm t,z)+ \xi(\bm t,w)}
e^{-\xi (\tilde\pd ,w^{-1}) -\xi (\tilde\pd ,z^{-1})} \tau_A(\bm t)\\
=& i_{w,z}\frac{w-z}{2(w+z)}\cdot
\exp\Big(\sum_{n>0:\text{ odd}}t_n(w^n+z^n) \Big)\\
&\cdot\exp\Big(-\sum_{n>0:\text{ odd}}\frac{2}{n}(w^{-n}+z^{-n})\frac{\pd}{\pd t_n}\Big).
\tau_A(\bm t),
\end{split}
\end{equation*}
where the operator $\exp\big(-\sum_{n>0:\text{ odd}}\frac{2}{n}(w^{-n}+z^{-n})\frac{\pd}{\pd t_n}\big)$
acts by shifting each time variable $t_n$ by $-(\frac{2}{n}(w^{-n}+z^{-n}))$.
Thus:
\begin{equation*}
\begin{split}
&\langle 0| e^{H_+(\bm t)} \phi(w)\phi(z) e^A|0\rangle\\
=&
i_{w,z}\frac{w-z}{2(w+z)}\cdot
\exp\Big(\sum_{n>0:\text{ odd}}t_n(w^n+z^n)\Big) \cdot
\tau_A \Big(t_n-\frac{2}{n}(w^{-n} +z^{-n})\Big).
\end{split}
\end{equation*}
Restricting to $\bm t=0$,
one obtains:
\begin{equation*}
\langle \phi(w)\phi(z) \rangle_A
= i_{w,z}\frac{w-z}{2(w+z)}\cdot
\tau_A (\bm t) \big|_{t_n = -\frac{2}{n}(w^{-n} +z^{-n})}.
\end{equation*}
Then by \eqref{eq-def-Ahat} we have:
\begin{Proposition}
\label{prop-gen-tau}
The generating series $\wA(w,z)$ of affine coordinates is given by the following
evaluation of the tau-function $\tau_A(t_1,t_3,t_5,\cdots)$:
\be
\wA(w,z)
= - i_{w,z}\frac{w-z}{4(w+z)}\cdot
\tau_A (\bm t) \big|_{t_n = -\frac{2}{n}(w^{-n} +z^{-n})}.
\ee
Or equivalently,
\be
A(w,z)
=  -i_{w,z}\frac{w-z}{4(w+z)}\cdot \Big(
\tau_A (\bm t) \big|_{t_n = -\frac{2}{n}(w^{-n} +z^{-n})} -1 \Big).
\ee
\end{Proposition}

\section{A Formula for Connected Bosonic $N$-Point Functions}
\label{sec-conn}

In the previous section we have computed the
bosonic $n$-point functions associated to a BKP tau-function $\tau_A(\bm t)$.
Now in this section,
we derive a formula for the connected bosonic $n$-point functions
$\langle H(z_1)\cdots H(z_n)\rangle_A^c$
in terms of the generating series of affine coordinates.
This formula is the BKP-analogue of the formula \eqref{eq-intro-Zhou} derived by Zhou in \cite{zhou1}.
We will see that the connected bosonic $n$-point functions $\langle H(z_1)\cdots H(z_n)\rangle_A^c$
is the generating series of the $n$-point correlators of
the free energy $F_A=\log \tau_A$
(with an additional modification at $n=2$).

\subsection{Connected bosonic $n$-point functions}

First we recall the notion of the connected bosonic $n$-point functions
associated to a tau-function $\tau_A(\bm t)$.

Following \cite[\S 5.1]{zhou1},
define the connected $n$-point functions
$\langle H(z_1)\cdots H(z_n)\rangle_A^c$
associated to the tau-function $\tau_A(\bm t)$ by
the M\"obius inversion formulas:
\be
\begin{split}
&\langle H(z_1)\cdots H(z_n) \rangle_A =
\sum_{I_1\sqcup \cdots \sqcup I_k=[n]}
\frac{1}{k!} \langle H(z_{I_1}) \rangle_A^c
\cdots \langle H(z_{I_k}) \rangle_A^c,\\
&\langle H(z_1)\cdots H(z_n) \rangle_A^c :=
\sum_{I_1\sqcup \cdots \sqcup I_k=[n]}
\frac{(-1)^{k-1}}{k} \langle H(z_{I_1}) \rangle_A
\cdots \langle H(z_{I_k}) \rangle_A,
\end{split}
\ee
where $[n]:=\{1,2,\cdots,n\}$;
and for a subset $I=\{i_1,i_2,\cdots,i_m\} \subset [n]$
with $i_1<i_2<\cdots <i_m$,
denote:
\begin{equation*}
H(z_I)= H(z_{i_1}) H(z_{i_2}) \cdots H(z_{i_m}).
\end{equation*}
For example:
\begin{equation*}
\begin{split}
&\langle H(z_1)\rangle_A^c = \langle H(z_1)\rangle_A,\\
&\langle H(z_1)H(z_2)\rangle_A^c = \langle H(z_1)H(z_2)\rangle_A
- \langle H(z_1)\rangle_A \cdot \langle H(z_2)\rangle_A,\\
&\langle H(z_1)H(z_2)H(z_3)\rangle_A^c = \langle H(z_1)H(z_2)H(z_3)\rangle_A
- \langle H(z_1) H(z_2)\rangle_A \cdot \langle H(z_3)\rangle_A\\
&\qquad\qquad\qquad\qquad
- \langle H(z_1) H(z_3)\rangle_A \cdot  \langle H(z_2)\rangle_A
- \langle H(z_2) H(z_3)\rangle_A \cdot  \langle H(z_1)\rangle_A\\
&\qquad\qquad\qquad\qquad
+2 \langle H(z_1)\rangle_A \cdot  \langle H(z_2)\rangle_A \cdot  \langle H(z_3)\rangle_A.
\end{split}
\end{equation*}

\subsection{Relation to the free energy}

Given a BKP tau-function $\tau_A=\tau_A(\bm t)$,
one can consider the formal quantum field theory associated to $\tau_A$
(see \cite{zhou5} for an introduction of the notion of formal quantum field theory).
Roughly speaking,
we regard the tau-function $\tau_A = \tau_A (\bm t)$ as a partition function
and regard the BKP-time variables $\bm t=(t_1,t_3,t_5,\cdots)$ as the coupling constants.
Then the logarithm
\be
F_A(\bm t) := \log \tau_A (\bm t)
\ee is called the free energy,
and the coefficient of the term $t_1^{m_1}t_3^{m_3}\cdots t_{2k+1}^{m_{2k+1}}$
(where $m_i\geq 0$ and $\sum_i m_i$ =n) in $F_A(\bm t)$
is called a connected $n$-point correlator.

Now we are interested in the computation of the free energy $F_A$,
or equivalently, the computation of the connected correlators.
This question is actually equivalent to the computation of the connected bosonic $n$-point functions
$\langle H(z_1)\cdots H(z_n)\rangle_A^c$,
since we have the following:
\begin{Lemma}
\label{lem-correlator}
For every $n\geq 1$,
\be
\label{eq-lem-correlator}
\begin{split}
&\sum_{i_1,\cdots,i_n> 0: \text{ odd}}
\frac{\pd^n F_A(\bm t)}{\pd t_{i_1}\cdots \pd t_{i_n}} \bigg|_{\bm t=0}
\cdot z_1^{-i_1}\cdots z_n^{-i_n}\\
= &
-\delta_{n,2} \cdot
i_{z_1,z_2}
\frac{z_1z_2(z_2^2+z_1^2)}{2(z_1^2-z_2^2)^2}
+\langle H(z_1)\cdots H(z_n)\rangle_A^c,
\end{split}
\ee
where:
\be
i_{z_1,z_2}
\frac{z_1z_2(z_2^2+z_1^2)}{2(z_1^2-z_2^2)^2}:
=\sum_{n>0:\text{ odd}}
\frac{n}{2}z_1^{-n}z_2^{n}.
\ee
\end{Lemma}
This lemma can be proved by the same method used by Zhou in \cite[\S 5]{zhou1},
and one only needs to replace the KP-time variables $(T_1,T_2,T_3,\cdots)$ in that work
by the BKP-time variables $(t_1,t_3,t_5,\cdots)$.
Here we briefly review the verification of the cases $n=1,2$
since our additional term
\begin{equation*}
-\delta_{n,2} \cdot
i_{z_1,z_2}
\frac{z_1z_2(z_2^2+z_1^2)}{2(z_1^2-z_2^2)^2}
\end{equation*}
appearing at $n=2$
is different from the additional term in \cite{zhou1}.
First, we denote:
\begin{equation*}
\begin{split}
f(z_1,\cdots,z_n):=& \sigma_B \big(H(z_1)\cdots H(z_n) e^A|0\rangle \big) /\tau_A(\bm t)\\
=& \langle 0 | e^{H_+(\bm t)} H(z_1)\cdots H(z_n) e^A|0\rangle /\tau_A(\bm t) ,
\end{split}
\end{equation*}
and define $f^c (z_1,\cdots,z_n)$ by:
\begin{equation*}
f^c (z_1,\cdots,z_n) :=
\sum_{I_1\sqcup \cdots \sqcup I_k=[n]}
\frac{(-1)^{k-1}}{k} f(z_{I_1})
\cdots f(z_{I_k}),
\end{equation*}
then we easily see that:
\begin{equation*}
\begin{split}
&\langle H(z_1)\cdots H(z_n)\rangle_A = f (z_1,\cdots,z_n)|_{\bm t=0},\\
&\langle H(z_1)\cdots H(z_n)\rangle_A^c = f^c (z_1,\cdots,z_n)|_{\bm t=0}.
\end{split}
\end{equation*}

Now we can compute $f(z_1,\cdots,z_n)$ using \eqref{eq-bfcor-boson}.
For $n=1$ we have:
\begin{equation*}
\begin{split}
f(z) =& \frac{1}{\tau_A(\bm t)}\cdot
\Big(
\sum_{n>0:\text{ odd}} \frac{\pd}{\pd t_n} \cdot z^{-n}
+\sum_{n>0:\text{ odd}} \frac{n}{2}t_n \cdot z^{n} \Big) \tau_A(\bm t)\\
=& \sum_{n>0:\text{ odd}} \frac{\pd F_A(\bm t)}{\pd t_n} \cdot z^{-n}
+\sum_{n>0:\text{ odd}} \frac{n}{2}t_n \cdot z^{n},
\end{split}
\end{equation*}
and thus by restricting to $\bm t=0$ we obtain a proof of the case $n=1$ of \eqref{eq-lem-correlator},
since $f^c(z) = f(z)$.
Similarly, for $n=2$ one has:
\begin{equation*}
f(z_1,z_2)
= \frac{1}{\tau_A}
\sum_{n_1,n_2>0:\text{ odd}} \Big(\frac{\pd}{\pd t_{n_1}}z_1^{-n_1} + \frac{n_1}{2}t_{n_1}z_1^{n_1}\Big)
\Big(\frac{\pd}{\pd t_{n_2}}z_2^{-n_2} + \frac{n_2}{2}t_{n_2}z_2^{n_2}\Big)
\tau_A,
\end{equation*}
and thus:
\begin{equation*}
\begin{split}
f(z_1,z_2)
=&\sum_{n_1,n_2>0:\text{ odd}} \Big(\frac{\pd^2 F_A}{\pd t_1 \pd t_2} + \frac{\pd F_A}{\pd t_1} \frac{\pd F_A}{\pd t_2}
\Big)z_1^{-n_1}z_2^{-n_2}
+ \sum_{n>0:\text{ odd}}
\frac{n}{2}z_1^{-n}z_2^{n}\\
&+ \sum_{n_1,n_2>0:\text{ odd}} \Big(
\frac{n_1}{2} t_{n_1} \frac{\pd F_A}{\pd t_{n_2}}
\cdot z_1^{n_1}z_2^{-n_2}
+ \frac{n_2}{2} t_{n_2} \frac{\pd F_A}{\pd t_{n_1}}
\cdot z_1^{-n_1}z_2^{n_2} \Big)\\
&+ \sum_{n_1,n_2>0:\text{ odd}} \Big(
\frac{n_1 n_2}{4} t_{n_1}t_{n_2}
\Big)z_1^{n_1}z_2^{n_2}\\
=& \sum_{n_1,n_2>0:\text{ odd}} \frac{\pd^2 F_A}{\pd t_1 \pd t_2} \cdot z_1^{-n_1}z_2^{-n_2}
+ f(z_1)f(z_2) + i_{z_1,z_2} \frac{z_1z_2(z_2^2+z_1^2)}{2(z_1^2-z_2^2)^2},
\end{split}
\end{equation*}
and then by $f^c(z_1,z_2) = f(z_1,z_2)-f(z_1)f(z_2)$ we get:
\begin{equation*}
f^c(z_1,z_2) = i_{z_1,z_2} \frac{z_1z_2(z_2^2+z_1^2)}{2(z_1^2-z_2^2)^2}+
\sum_{n_1,n_2>0:\text{ odd}} \frac{\pd^2 F_A}{\pd t_1 \pd t_2} \cdot z_1^{-n_1}z_2^{-n_2}.
\end{equation*}
Taking $\bm t=0$, and we have proved the case $n=2$.

In general cases $n\geq 3$,
the relation \eqref{eq-lem-correlator}
follows from:
\begin{equation*}
f^c (z_1,\cdots,z_n) =
\sum_{i_1,\cdots,i_n> 0: \text{ odd}}
\frac{\pd^n F_A(\bm t)}{\pd t_{i_1}\cdots \pd t_{i_n}}
\cdot z_1^{-i_1}\cdots z_n^{-i_n},
\qquad \forall n\geq 3.
\end{equation*}
See \cite[Prop. 5.1]{zhou1} for a detailed proof for $n\geq 3$,
and here we will not repeat this.

\subsection{Computation of the connected bosonic $n$-point functions}
\label{sec-conn-npt}

In this subsection
we derive a formula for the connected bosonic $n$-point functions
of a tau-function $\tau_A$
using the results in \S \ref{sec-bos-npt}.
First we prove a combinatorial result about Pfaffians.
The following is a Pfaffian-analogue of \cite[Prop. 5.2]{zhou1}:
\begin{Proposition}
Assume $\xi(x,y)$ is a function with $\xi(x,y)=-\xi(y,x)$,
and for each $n\geq 1$ we define an anti-symmetric matrix $M(n)$ of size $2n\times 2n$ by:
\be
M(n)_{ij} =
\xi\big((-1)^i z_{\lceil \frac{i}{2} \rceil }, (-1)^j z_{\lceil \frac{j}{2} \rceil })
\ee
for $1\leq i<j\leq 2n$.
Define:
\begin{equation*}
\varphi(z_1,\cdots,z_n):= \Pf(M(n)_{ij})_{1\leq i,j \leq 2n}
\end{equation*}
for every $n$,
then the connected version
\begin{equation*}
\varphi^c (z_1,\cdots,z_n) :=
\sum_{I_1\sqcup \cdots \sqcup I_k=[n]}
\frac{(-1)^{k-1}}{k} \varphi(z_{I_1})
\cdots \varphi(z_{I_k}),
\end{equation*}
is given by:
\be
\label{eq-prop-cycle}
\varphi^c (z_1,\cdots,z_n) = \sum_{\substack{\text{$n$-cycles $\sigma$} \\ \epsilon_2,\cdots,\epsilon_n \in\{\pm 1\}}}
(-\epsilon_2\cdots\epsilon_n) \cdot
\prod_{i=1}^n \xi(\epsilon_{\sigma(i)} z_{\sigma(i)}, -\epsilon_{\sigma(i+1)} z_{\sigma(i+1)}),
\ee
where we use the conventions
$\epsilon_{1} :=1$ and
$\sigma(n+1):=\sigma(1)$.
\end{Proposition}
\begin{proof}
We prove by induction on $n$.
For $n=1$,
there is only one $1$-cycle $\sigma=(1)$,
and the right-hand side of \eqref{eq-prop-cycle} is $-\xi(z_1,-z_1) = \xi(-z_1,z_1)$.
This matches with:
\begin{equation*}
\varphi^c(z_1) := \varphi(z_1) =
\Pf \left[
\begin{array}{cc}
0 & \xi(-z_1,z_1)\\
-\xi(-z_1,z_1) & 0\\
\end{array}
\right]
= \xi(-z_1,z_1).
\end{equation*}

Now assume \eqref{eq-prop-cycle} holds for $1,2,\cdots,n-1$,
and consider the case of $n$.
We introduce some notations for convenience.
Let $C_k\subset S_n$ be the subset of permutations that can be decomposed as
a product of $k$ cycles.
Given $\sigma\in C_k$,
one decomposes it as a product
$\sigma = \sigma_1\cdots\sigma_k$,
where $\sigma_j$ is a cycle
$\sigma_j = (i_j^1 i_j^2 \cdots i_j^{r_j})$.
Denote:
\begin{equation*}
X(\bm z;\sigma)= \prod_{j=1}^k
\sum_{\substack{ \epsilon(i_j^l)\in \{\pm 1\} \\ l=2,\cdots,r_j}}
\Big(-\prod_{l=2}^{r_j} \epsilon(i_j^l) \Big)
\prod_{s=1}^{r_j} \xi(\epsilon(\sigma_j(i_j^s))z_{\sigma_j(i_j^s)},
-\epsilon(\sigma_j(i_j^{s+1}))z_{\sigma_j(i_j^{s+1})}),
\end{equation*}
where we use the convention $\epsilon_j^1 :=1$ and $\sigma_j(i_j^{r_j+1}) :=\sigma_j(i_j^1)$.

Now recall that by the M\"obius inversion formula we have:
\be
\label{eq-pfmain-1}
\begin{split}
\varphi(z_1,\cdots,z_n) =&
\sum_{I_1\sqcup \cdots \sqcup I_k=[n]}
\frac{1}{k!} \varphi^c(z_{I_1})\cdots \varphi^c(z_{I_k})\\
=& \varphi^c(z_1,\cdots,z_n)
+\sum_{k\geq 2}
\sum_{I_1\sqcup \cdots \sqcup I_k=[n]}
\frac{1}{k!}
\varphi^c(z_{I_1})\cdots \varphi^c(z_{I_k}).
\end{split}
\ee
Let $[n]=I_1\sqcup\cdots \sqcup I_k$ be a decomposition of $[n]=\{1,2,\cdots,n\}$ (for $k\geq 2$),
and denote $I_j = \{i_j^1,\cdots,i_j^{|I_j|}\}$ where $i_j^1<\cdots<i_j^{|I_j|}$.
Then by induction hypothesis:
\begin{equation*}
\varphi^c (z_{I_j})=
\sum_{\substack{\sigma_j :\text{ $|I_j|$-cycle}\\ \epsilon_j^2,\cdots,\epsilon_j^{|I_j|}\in \{\pm 1\}}}
(- \epsilon_j^2\cdots\epsilon_j^{|I_j|})
\prod_{s=1}^{|I_j|} \xi(\epsilon_j^{\sigma_j(s)}z_{i_j^{\sigma_j(s)}},
-\epsilon_j^{\sigma_j(s+1)}z_{i_j^{\sigma_j(s+1)}}),
\end{equation*}
and thus:
\begin{equation*}
\varphi^c (z_{I_1})\cdots \varphi^c (z_{I_k})= (-1)^k
\sum_{\substack{\sigma_1,\cdots,\sigma_k \\ \epsilon_j^l\in \{\pm 1\}}}
\prod_{j=1}^k \prod_{l=1}^{|I_j|} \epsilon_j^l
\prod_{s=1}^{|I_j|} \xi(\epsilon_j^{\sigma_j(s)}z_{i_j^{\sigma_j(s)}},
-\epsilon_j^{\sigma_j(s+1)}z_{i_j^{\sigma_j(s+1)}}).
\end{equation*}
Now fix a decomposition $[n]=I_1\sqcup\cdots \sqcup I_k$.
Then given a family of cycles $\sigma_1,\cdots,\sigma_k$,
one can associate a permutation $\sigma\in S_n$ by:
\begin{equation*}
\sigma = (i_1^{\sigma_1(1)}\cdots i_1^{\sigma_1(|I_1|)})
\cdots (i_k^{\sigma_k(1)}\cdots i_k^{\sigma_k(|I_k|)}).
\end{equation*}
Conversely,
each permutation $\sigma\in S_n$ can be uniquely decomposed into a product of cycles.
Therefore,
from the above discussions we see:
\be
\label{eq-pfmain-2}
\sum_{\sigma\in C_k} X(z;\sigma) =
\sum_{I_1\sqcup \cdots \sqcup I_k=[n]}
\frac{1}{k!}
\varphi^c(z_{I_1})\cdots \varphi^c(z_{I_k}),
\qquad \forall k\geq 2,
\ee
where the additional factor $\frac{1}{k!}$ indicates that
there are $k!$ ways to permute the indices of the subsets $I_1,\cdots,I_k$.

Recall that the conclusion we want to prove is actually:
\begin{equation*}
\varphi^c(z_1,\cdots,z_n) =
\sum_{\sigma\in C_1} X(\bm z;\sigma),
\end{equation*}
thus by \eqref{eq-pfmain-1} and \eqref{eq-pfmain-2}, it suffices to prove:
\be
\label{eq-pfmain-3}
\varphi(z_1,\cdots,z_n) = \sum_{\sigma\in S_n} X(\bm z;\sigma).
\ee
Now we prove this equality.
First recall that:
\be
\label{eq-mainpf-4}
\begin{split}
\varphi(z_1,\cdots,z_n) = &
\sum_{\substack{(p_1,q_1,\cdots,p_n,q_n)\\p_i<q_i,\quad p_1<\cdots<p_n}}
\text{sgn(p,q)}\cdot \prod_{i=1}^n (M(n))_{p_iq_i}\\
=& \sum_{\substack{(p_1,q_1,\cdots,p_n,q_n)\\p_i<q_i,\quad p_1<\cdots<p_n}}
\text{sgn(p,q)}\cdot \prod_{i=1}^n
\xi\big((-1)^{p_i} z_{\lceil \frac{p_i}{2} \rceil }, (-1)^{q_i} z_{\lceil \frac{q_i}{2} \rceil }
\big),
\end{split}
\ee
where $(p_1,q_1,\cdots,p_n,q_n)$ is a permutation of $(1,2,\cdots,2n)$.
Notice that the arguments
$\{(-1)^{p_i} z_{\lceil \frac{p_i}{2} \rceil }, (-1)^{q_j} z_{\lceil \frac{q_j}{2} \rceil } \}$
run over $\{\pm z_1,\cdots,\pm z_n\}$,
thus one can always decompose the product
\be
\label{eq-mainpf-product}
\prod_{i=1}^n
\xi\big((-1)^{p_i} z_{\lceil \frac{p_i}{2} \rceil }, (-1)^{q_i} z_{\lceil \frac{q_i}{2} \rceil }
\big)
\ee
into some `loops' of the form:
\begin{equation*}
\pm \xi(z_{i_1},-\gamma_2 z_{i_2})\xi(\gamma_2 z_{i_2},-\gamma_3 z_{i_3})
\cdots \xi(\gamma_r z_{i_r},- z_{i_1}),
\end{equation*}
where $\gamma_2,\cdots,\gamma_r = \pm 1$,
and we have used $\xi(x,y) = -\xi(y,x)$ to rearrange the arguments suitably
and produced some factors $\pm1$.
A loop of this form determines a cycle of length $r$ in $S_n$,
and a decomposition into such loops corresponds to an element $\sigma$ in $S_n$
which is the product of these cycles $\sigma=(i_1\cdots i_r)(j_1 \cdots j_s)\cdots$.
One easily sees that for each cycle $(i_1\cdots i_r)$,
the $(r-2)$ signs $\gamma_2,\cdots,\gamma_r = \pm1$ can be chosen arbitrarily.
Moreover,
a choice of the permutation $(p_i,q_i)_{i=1}^n$ of $[2n]$
(with $p_i<q_i$ and $p_1<\cdots p_n$) is equivalent to
a choice of the permutation $\sigma \in S_n$ together with this signs $\pm 1$ for each cycle.
Thus we obtain:
\begin{equation*}
\varphi(z_1,\cdots,z_n) = \sum_{\sigma\in S_n} Y(\bm z;\sigma),
\end{equation*}
where $Y(\bm z;\sigma)$ is of the form:
\begin{equation*}
Y(\bm z;\sigma)= \pm \prod_{j=1}^k
\sum_{\substack{ \gamma(i_j^l)\in \{\pm 1\} \\ l=2,\cdots,r_j}}
\prod_{l=2}^{r_j} \gamma(i_j^l)
\prod_{s=1}^{r_j} \xi(\gamma(\sigma(i_j^s))z_{\sigma(i_j^s)},
-\gamma(\sigma(i_j^{s+1}))z_{\sigma(i_j^{s+1})}),
\end{equation*}
and the choice of the sign $\pm$ is determined by the
rearrangement of the product \eqref{eq-mainpf-product} using
$\xi(x,y) = -\xi(y,x)$.
Notice that fixing a permutation $\sigma\in S_n$ is actually equivalent to
fixing a family of cycles $\sigma_1,\cdots,\sigma_k$
(and there are $k!$ ways to permuting the subscripts),
thus by comparing $X(\bm z;\sigma)$ with $Y(\bm z;\sigma)$
we see that in order to prove \eqref{eq-pfmain-3},
we only need to show the sign $\pm$ in $Y(\bm z;\sigma)$ is $(-1)^k$ for $\sigma\in C_k$.

First we consider a permutations $\sigma^0\in S_n$ of the following standard form:
\be
\label{eq-pfmain-standard}
\sigma^0 = (1,\cdots, r_1)(r_1+1,\cdots, r_2)\cdots (n-r_k+1,\cdots, n),
\ee
together with the simplest choice of signs $\gamma_i \equiv 1$.
In this case the sign $\pm$ is simply $\text{sgn}(p,q)$.
Notice that the first cycle $(1,2,\cdots,r_1)$ corresponds to the product:
\begin{equation*}
\xi(z_1,-z_2)\xi(z_2,-z_3)\cdots\xi(z_{r_1},-z_1),
\end{equation*}
and the corresponding permutation $(p_i,q_i)$ of $[2n]$ contains the cycle
$(2,3,\cdots,2r_1,1)$ which is an odd permutation in $S_{2n}$.
Similarly,
every cycle in $\sigma^0$ corresponds to an odd cycle in $(p_i,q_i)$,
therefore $\text{sgn}(p,q) = (-1)^k$ where $k$ is the number of cycles,
which proves the conclusion in this special case.

Now we consider a general permutation $\sigma\in S_n$ together with the simplest choice
of signs $\gamma_i \equiv 1$.
Assume that $\sigma$ contains $k$ cycles,
then it is conjugate to the standard form \eqref{eq-pfmain-standard} for some $r_i$,
i.e.,
one can find a sequence of transpositions $\tau_1,\cdots,\tau_t \in S_n$ such that
$\sigma = \tau_t\cdots\tau_1 \sigma^0 \tau_1^{-1}\cdots\tau_t^{-1}$.
It is easy to check that conjugation by a transposition $\tau_i$
does not change $\text{sgn}(p,q)$ and the sign produced in rearranging
the produce \eqref{eq-mainpf-product}.
Furthermore,
replacing a sign $\gamma_i=1$ by $\gamma_i=-1$ also preserves the sign $\pm $ in $Y(\bm z;\sigma)$.
Thus the sign $\pm $ in $Y(\bm z;\sigma)$ is always $(-1)^k$ where $k$ is
the number of cycles in $\sigma\in S_n$,
thus we have proved $Y(\bm z;\sigma) = X(\bm z;\sigma)$.
\end{proof}

Now we can state our main result of this work.
Take the matrix $M(n)$ to be $(C_{ij})_{1\leq i,j\leq 2n}$
(see Theorem \ref{thm-bos-npt}),
then we find a way to compute the connected $n$-point functions $\langle H(z_1)\cdots H(z_n)\rangle_A^c$.
For $n=1$,
this simply tell us:
\begin{equation*}
\langle H(z_1) \rangle_A^c =\xi(-z_1,z_1) =A(-z_1,z_1).
\end{equation*}
Notice that the special case $\sigma(i)=\sigma(i+1)$ only appearing in the case $n=1$ since $\sigma$ is an $n$-cycle.

And for $n\geq 2$,
the above proposition gives the following formula for
the connected bosonic $n$-point functions:
\begin{equation*}
\langle H(z_1)\cdots H(z_n) \rangle_A^c =
\sum_{\substack{\text{$n$-cycles $\sigma$} \\ \epsilon_2,\cdots,\epsilon_n \in\{\pm 1\}}}
(-\epsilon_2\cdots\epsilon_n) \cdot
\prod_{i=1}^n \xi(\epsilon_{\sigma(i)} z_{\sigma(i)}, -\epsilon_{\sigma(i+1)} z_{\sigma(i+1)}),
\end{equation*}
where for $\sigma(i)<\sigma(i+1)$,
\be
\label{eq-def-xiA-1}
\xi(\epsilon_{\sigma(i)} z_{\sigma(i)}, -\epsilon_{\sigma(i+1)} z_{\sigma(i+1)}) :=
\wA (\epsilon_{\sigma(i)} z_{\sigma(i)}, -\epsilon_{\sigma(i+1)} z_{\sigma(i+1)});
\ee
and for $\sigma(i)>\sigma(i+1)$,
\be
\label{eq-def-xiA-2}
\begin{split}
\xi(\epsilon_{\sigma(i)} z_{\sigma(i)}, -\epsilon_{\sigma(i+1)} z_{\sigma(i+1)}) :=&
-\xi (-\epsilon_{\sigma(i+1)} z_{\sigma(i+1)} ,\epsilon_{\sigma(i)} z_{\sigma(i)})\\
=& -\wA( -\epsilon_{\sigma(i+1)} z_{\sigma(i+1)} ,\epsilon_{\sigma(i)} z_{\sigma(i)}).
\end{split}
\ee
Thus by Lemma \ref{lem-correlator},
we obtain the following:
\begin{Theorem}
\label{thm-main-conn}
Let $\tau_A = \tau(\bm t)$ be a tau-function of the BKP hierarchy with $\tau(0)=1$,
and let $A(w,z),\wA(w,z)$ be the generating series of the affine coordinates
defined by \eqref{eq-def-seriesA} and \eqref{eq-def-Ahat} respectively.
Denote $F_A=\log \tau_A$,
then we have:
\be
\sum_{i> 0: \text{ odd}}
\frac{\pd F_A(\bm t)}{\pd t_{i}} \bigg|_{\bm t=0}
\cdot z^{-i}
=A(-z,z),
\ee
and for $n\geq 2$,
\be
\begin{split}
&\sum_{i_1,\cdots,i_n> 0: \text{ odd}}
\frac{\pd^n F_A(\bm t)}{\pd t_{i_1}\cdots \pd t_{i_n}} \bigg|_{\bm t=0}
\cdot z_1^{-i_1}\cdots z_n^{-i_n}
=
-\delta_{n,2} \cdot
i_{z_1,z_2}
\frac{z_1z_2(z_2^2+z_1^2)}{2(z_1^2-z_2^2)^2}\\
&\qquad\qquad\qquad
+ \sum_{\substack{ \sigma: \text{ $n$-cycle} \\ \epsilon_2,\cdots,\epsilon_n \in\{\pm 1\}}}
(-\epsilon_2\cdots\epsilon_n) \cdot
\prod_{i=1}^n \xi(\epsilon_{\sigma(i)} z_{\sigma(i)}, -\epsilon_{\sigma(i+1)} z_{\sigma(i+1)}),
\end{split}
\ee
where $\xi$ is given by \eqref{eq-def-xiA-1}-\eqref{eq-def-xiA-2},
and we use the conventions
$\epsilon_{1} :=1$,
$\sigma(n+1):=\sigma(1)$.
\end{Theorem}

\begin{Example}
Here we write down the explicit formulas for small $n$.
For $n=2$,
there is only one $2$-cycle $\sigma=(12)$,
thus:
\begin{equation*}
\begin{split}
&\sum_{n_1,n_2>0:\text{ odd}}
\frac{\pd^2 F_A (\bm t)}{\pd t_{n_1} \pd t_{n_2}}\bigg|_{\bm t=0} z_1^{-n_1}z_2^{-n_2}\\
=& -i_{z_1,z_2}\frac{z_1z_2(z_2^2+z_1^2)}{2(z_1^2-z_2^2)^2}
-\widehat A(z_1,z_2) \widehat A(-z_1,-z_2)
+ \widehat A(z_1,-z_2) \widehat A(-z_1,z_2).
\end{split}
\end{equation*}
And for $n=3$,
there are two $3$-cycles $\sigma=(123),(132)$,
thus the result is:
\begin{equation*}
\begin{split}
&\sum_{n_1,n_2,n_3>0:\text{ odd}}
\frac{\pd^3 F_A(\bm t)}{\pd t_{n_1} \pd t_{n_2} \pd t_{n_3}}\bigg|_{\bm t=0}
z_1^{-n_1}z_2^{-n_2} z_3^{-n_3}
\\
=&
- \wA(z_1,z_2)\wA(-z_2,-z_3)\wA(-z_1,z_3) + \wA(z_1,z_2)\wA(-z_2,z_3)\wA(-z_1,-z_3)\\
& - \wA(z_1,-z_2)\wA(z_2,z_3)\wA(-z_1,-z_3) + \wA(z_1,-z_2)\wA(z_2,-z_3)\wA(-z_1,z_3) \\
& + \wA(z_1,z_3)\wA(-z_2,-z_3)\wA(-z_1,z_2) - \wA(z_1,z_3)\wA(z_2,-z_3)\wA(-z_1,-z_2)\\
& + \wA(z_1,-z_3)\wA(z_2,z_3)\wA(-z_1,-z_2) - \wA(z_1,-z_3)\wA(-z_2,z_3)\wA(-z_1,z_2) .
\end{split}
\end{equation*}
For $n=4$,
there are six $4$-cycles $\sigma=(1234),(1243),(1324),(1342),(1423),(1432)$,
and there are $6\times 2^3 =48$ terms in the right-hand side of the formula.
For simplicity here we only write down the first $8$ terms
(corresponding to $\sigma=(1234)$):
\begin{equation*}
\begin{split}
&\sum_{n_1,n_2,n_3,n_4>0:\text{ odd}} \frac{\pd^4 F_A(\bm t)}
{\pd t_{n_1} \pd t_{n_2} \pd t_{n_3} \pd t_{n_4}}\bigg|_{\bm t=0}
z_1^{-n_1}z_2^{-n_2}z_3^{-n_3}z_4^{-n_4}
\\
=& \wA(z_1,-z_2) \wA(z_2,-z_3)\wA(z_3,-z_4) \wA(-z_1,z_4) \\
&- \wA(z_1,-z_2) \wA(z_2,-z_3)\wA(z_3,z_4) \wA(-z_1,-z_4)\\
&- \wA(z_1,-z_2) \wA(z_2,z_3)\wA(-z_3,-z_4) \wA(-z_1,z_4)\\
&+ \wA(z_1,-z_2) \wA(z_2,z_3)\wA(-z_3,z_4) \wA(-z_1,-z_4)\\
&- \wA(z_1,z_2) \wA(-z_2,-z_3)\wA(z_3,-z_4) \wA(-z_1,z_4)\\
&+ \wA(z_1,z_2) \wA(-z_2,-z_3)\wA(z_3,z_4) \wA(-z_1,-z_4)\\
&+ \wA(z_1,z_2) \wA(-z_2,z_3)\wA(-z_3,-z_4) \wA(-z_1,z_4)\\
&- \wA(z_1,z_2) \wA(-z_2,z_3)\wA(-z_3,z_4) \wA(-z_1,-z_4)
+\text{other $40$ terms}.
\end{split}
\end{equation*}
The other $40$ terms are obtained by permuting the indices $\{2,3,4\}$,
suitably exchange the arguments in $\wA$ such that $i<j$ in $\wA(\pm z_i, \pm z_j)$,
and then suitably multiplying by some $\pm 1$ on each term.
\end{Example}

\section{Tau-Functions of KdV Hierarchy: KP vs. BKP}
\label{sec-kdv}

In \cite{al1},
Alexandrov showed that if $\tau = \tau(\bm t)$ is a tau-function
of the KdV hierarchy,
where $\bm t = (t_1,t_3,t_5,\cdots)$ are the KdV-time variables,
then
\begin{equation*}
\tilde\tau (\bm t):= \tau(\bm t/2)
\end{equation*}
is a tau-function of the BKP hierarchy with time variables $\bm t$.
Moreover,
it is well-known that the KdV hierarchy is a reduction of the KP hierarchy,
thus $\tau$ is automatically a tau-function of the KP hierarchy.
Now given a tau-function $\tau$ of the KdV hierarchy,
one has two parallel approaches to study its affine coordinates
and compute the connected $n$-point functions:
\begin{itemize}
\item
[1)]The KP-affine coordinates $\{a_{n,m}^{\text{KP}}\}_{n,m\geq 0}$ of $\tau$;
\item
[2)]The BKP-affine coordinates $\{a_{n,m}^{\text{BKP}}\}_{n,m\geq 0}$ of $\tilde\tau$.
\end{itemize}
In this section,
we show that the two generating series of these two family of affine coordinates
are related by a simple relation (see Theorem \ref{thm-KPBKP}).

\subsection{Relation between KP- and BKP-affine coordinates}

Let $\tau(\bm t)$ be a tau-function of the KdV hierarchy with $\tau(0)=1$,
and let $\tilde\tau (\bm t):= \tau(\bm t/2)$.
To avoid confusions,
we denote by $A^{\text{BKP}}(w,z)$ the generating seires \eqref{eq-def-seriesA} of affine coordinates
$\{a_{n,m}^{\text{BKP}}\}_{n,m\geq 0}$ of the BKP tau-function $\tilde\tau(\bm t)$,
and denote by
\be
A^{\text{KP}} (x,y) = \sum_{n,m\geq 0} a_{n,m}^{\text{KP}} x^{-n-1}y^{-m-1}
\ee
the generating series of the affine coordinates of the KP tau-function $\tau(\bm T)$
(see \cite{zhou1} for details).
Notice here $\bm T = (T_1,T_2,T_3,\cdots)$ are the KP-times variables,
and $t_n = T_n$ for every $n>0$ odd.
The KP tau-function $\tau(\bm t)=\tau(\bm T)$ is independent of $(T_2,T_4,T_6,\cdots)$
by the definition of the KdV hierarchy.

In \cite[\S 5.6]{zhou1},
Zhou has proved the following relation
for KP tau-functions:
\be
A^{\text{KP}} (x,y) = i_{x,y} \frac{1}{x-y} \Big(\tau(\bm T)\big|_{T_n= \frac{1}{n}(y^{-n}-x^{-n})}-1\Big),
\ee
thus in our case we have:
\be
\label{eq-KPBKP-1}
A^{\text{KP}} (x,-y) = i_{x,y} \frac{1}{x+y} \Big(\tau(\bm t)\big|_{t_n= -\frac{1}{n}(x^{-n}+y^{-n})}-1\Big),
\ee
since $\tau$ is independent of $(T_2,T_4,\cdots)$.

On the other hand,
by Proposition \ref{prop-gen-tau} we have:
\be
\label{eq-KPBKP-2}
\begin{split}
A^{\text{BKP}} (w,z)
=&  -i_{w,z}\frac{w-z}{4(w+z)}\cdot \Big(
\tilde\tau (\bm t) \big|_{t_n = -\frac{2}{n}(w^{-n} +z^{-n})} -1 \Big)\\
=& -i_{w,z}\frac{w-z}{4(w+z)}\cdot \Big(
\tau (\bm t) \big|_{t_n = -\frac{1}{n}(w^{-n} +z^{-n})} -1 \Big).
\end{split}
\ee
Now comparing \eqref{eq-KPBKP-1} with \eqref{eq-KPBKP-2},
we obtain the following:
\begin{Theorem}
\label{thm-KPBKP}
Let $\tau(\bm t)$ be a tau-function of the KdV hierarchy.
Then:
\be
A^{\text{BKP}} (w,z) = -\frac{w-z}{4} \cdot A^{\text{KP}}(w,-z).
\ee
\end{Theorem}

\subsection{Affine coordinates for a special class of KdV tau-functions}

\label{sec-phi12}

In \cite[\S 6.9]{zhou1} Zhou derived a simple formula
for the generating series of the KP-affine coordinates
of the Witten-Kontsevich tau-function in terms of the Faber-Zagier series,
using a result of Balogh-Yang in \cite{by}.
One easily sees that this method applies to a family of tau-functions
of the KdV hierarchy satisfies a special condition
(the condition $\det G(z)=1$ below).
This will provide us a way to find simple formulas for
the generating series of KP- and BKP-affine coordinates
of such KdV tau-functions.
In this subsection,
we first give a brief review of this method for such tau-functions,
and then combine it with Theorem \ref{thm-KPBKP}.

In Sato's theory \cite{sa},
a tau-function of the KdV hierarchy corresponds to
an element $W$ in the big cell of the Sato-Grassmannian satisfying the condition $z^2 W \subset W$.
Here $W$ is a subspace
$W\subset \bC[z]\oplus z^{-1}\bC[[z^{-1}]]$ of the form:
\begin{equation*}
W= \text{span} \{\Phi_1(z),\Phi_2(z),\Phi_3(z),\cdots \},
\end{equation*}
where
\begin{equation*}
\Phi_i(z) = z^{i-1} + \text{lower order terms}
\end{equation*}
is Laurent series in $z$ of degree $i-1$.
The subspace $W$ is uniquely determined by the first two basis vectors
$\Phi_1(z)$ and $\Phi_2(z)$,
since one has:
\begin{equation*}
W= \text{span} \{ z^{2n}\Phi_1(z), z^{2n}\Phi_2(z) \}_{n\geq 0}.
\end{equation*}
Denote:
\be
\label{eq-phi-12}
\Phi_1(z) = 1+\sum_{n\geq 1} a_n z^{-n},
\qquad
z^{-1}\Phi_2(z) = 1+\sum_{n\geq 1} b_n z^{-n},
\ee
and denote by $\{a_{n,m}^{\text{KP}}\}_{n,m\geq 0}$ the KP-affine coordinates
of this tau-function.
The following formula is due to Balogh-Yang \cite[Lemma 2.4]{by}:
\be
\label{eq-by-general}
\sum_{m,n\geq 0} \left[
\begin{array}{cc}
a_{2n,2m+1}^{\text{KP}} & a_{2n+1,2m+1}^{\text{KP}}\\
a_{2n,2m}^{\text{KP}} & a_{2n+1,2m}^{\text{KP}}
\end{array}
\right] x^{-m-1}y^{-n-1}=
\frac{1}{x-y} \big(I - G(x)G(y)^{-1}\big),
\ee
where $G$ is the matrix:
\begin{equation*}
G(z)= \left[
\begin{array}{cc}
1+\sum_{n\geq 1} a_{2n}z^{-n} & \sum_{n\geq 0} b_{2n+1}z^{-n}\\
\sum_{n\geq 1}a_{2n-1} z^{-n} & 1+\sum_{n\geq 1} b_{2n} z^{-n}
\end{array}
\right].
\end{equation*}
If the condition $\det (G(z)) =1$ holds,
then:
\begin{equation*}
G(y)^{-1}= \left[
\begin{array}{cc}
 1+\sum_{n\geq 1} b_{2n} z^{-n} & -\sum_{n\geq 0} b_{2n+1}z^{-n}\\
-\sum_{n\geq 1}a_{2n-1} z^{-n} & 1+\sum_{n\geq 1} a_{2n}z^{-n}
\end{array}
\right].
\end{equation*}
The relation \eqref{eq-by-general} for $2\times2$ matrices gives us four identities,
and the first one is:
\begin{equation*}
\begin{split}
&\sum_{m,n\geq 0}a_{2n, 2m+1}^{\text{KP}} x^{-m-1}y^{-n-1}\\
=&\frac{1}{x-y}\Big(1-\sum_{k,l\geq 0} a_{2k}x^{-k} b_{2l}y^{-l}
+ \sum_{k,l\geq 0} b_{2k+1}x^{-k} a_{2l-1}y^{-l}
\Big),
\end{split}
\end{equation*}
and taking $x=z^2,y=w^2$ gives:
\begin{equation*}
\begin{split}
&\sum_{m,n\geq 0}a_{2n, 2m+1}^{\text{KP}} z^{-2m-2}w^{-2n-1}\\
=&\frac{w}{z^2-w^2}\Big(1-\sum_{k,l\geq 0} a_{2k}z^{-2k} b_{2l}w^{-2l}
+ \sum_{k,l\geq 0} b_{2k+1}z^{-2k} a_{2l-1}w^{-2l}
\Big).
\end{split}
\end{equation*}
One can similarly write down the other three identities
and then sum the four identities together,
and the final result is
(see \cite[(282)]{zhou1} for the case of the Witten-Kontsevich tau-function):
\be
\sum_{m,n\geq 0}a_{n,m}^{\text{KP}}w^{-n-1} z^{-m-1}
=i_{w,z} \frac{1}{z-w} + \frac{\Phi_1(z)\Phi_2(-w)-\Phi_2(z)\Phi_1(-w)}{z^2-w^2}.
\ee
Thus by Theorem \ref{thm-KPBKP},
we obtain the following formulas
for the generating series of the BKP-affine coordinates $A^{\text{BKP}}$, $\wA^{\text{BKP}}$
defined by \eqref{eq-def-seriesA} and \eqref{eq-def-Ahat} respectively:
\begin{Proposition}
\label{prop-phi12}
Let $\tau(\bm t)$ be a tau-function of the KdV hierarchy satisfying the condition $\det G(z)=1$.
Then the generating series of the BKP-affine coordinates
$\{a_{m,n}^{\text{BKP}}\}$ of
$\tilde\tau(\bm t) = \tau(\bm t/2)$ is given by:
\be
A^{\text{BKP}}(w,z)
= \frac{w-z+\Phi_1(-z)\Phi_2(-w)-\Phi_1(-w)\Phi_2(-z)}{4(w+z)},
\ee
and
\be
\wA^{\text{BKP}}(w,z)
= \frac{\Phi_1(-z)\Phi_2(-w)-\Phi_1(-w)\Phi_2(-z)}{4(w+z)},
\ee
where $\Phi_1(z),\Phi_2(z)$ are the first two basis vectors \eqref{eq-phi-12}
of the corresponding point in the Sato-Grassmannian.
\end{Proposition}

Now one can plug $A^{\text{BKP}}(w,z)$ and $\wA^{\text{BKP}}(w,z)$ into
Theorem \ref{thm-main-conn} to obtain formulas for the connected $n$-point functions
and compute the free energy.
Here for the special case $n=1$, one needs to use L'H\^{o}pital's rule:
\begin{equation*}
\begin{split}
\sum_{n>0:\text{ odd}}
\frac{\pd \log\tilde\tau(\bm t)}{\pd t_n}\bigg|_{\bm t=0}\cdot z^{-n}
=&\lim_{w\to -z} A^{\text{BKP}}(w,z)\\
=& \frac{1-\Phi_1(-z)\Phi_2'(z)+\Phi_1'(z)\Phi_2(-z)}{4}.
\end{split}
\end{equation*}

In the next section,
we will present two examples
of tau-functions of this type.

\section{Examples: Witten-Kontsevich and BGW Tau-Functions}
\label{sec-WK-BGW}

Hypergeometric tau-functions (see Orlov \cite{or})
provide a large family of BKP tau-functions,
and they are known to be related to the spin Hurwitz numbers \cite{mmn}.
A hypergeometric tau-function is of the form:
\be
\tau_{\theta, \bm t^*} (\bm t/2)
=\sum_{\lambda\in DP}
2^{-l(\lambda)} \theta_\lambda
Q_\lambda (\bm t^*/2) Q_\lambda(\bm t/2),
\ee
where $\theta = \prod_{j=1}^l \prod_{k=1}^{\lambda_j} \theta(k)$ and $\theta(k)$
is a function on the set of positive integers,
and $\bm t^* = (t_1^*,t_3^*,t_5^*,\cdots)$.
The BKP-affine coordinates of $\tau_{\theta, \bm t^*} (\bm t/2)$
are:
\be
a_{0,n} = \frac{\theta_{(n)}}{2} Q_{(n)} (\bm t^* /2),\qquad
a_{n,m} = \frac{\theta_{(m,n)}}{4} Q_{(m,n)} (\bm t^* /2),
\ee
where $m,n>0$.
One can use Theorem \ref{thm-main-conn} to compute $\log\tau_{\theta, \bm t^*} (\bm t/2)$.

In what follows,
we will discuss two examples of hypergeometric tau-functions --
the Witten-Kontsevich tau-function and the BGW tau-function.
Using the recent results \cite{ly1,ly2} of X. Liu and the second author,
we obtain automatically the explicit formulas for their affine coordinates.
Moreover,
we can write down simple expressions for the generating series
using Proposition \ref{prop-phi12}.

\subsection{Affine coordinates of Witten-Kontsevich tau-function}

The famous Witten Conjecture/Kontsevich Theorem \cite{wit, kon} claims that
the generating series of intersection numbers of
$\psi$-classes on the Deligne-Mumford moduli space $\Mbar_{g,n}$ of stable curves \cite{dm, kn}
is a tau-function $\tau_{\text{WK}}$ of the KdV hierarchy.
Then
\be
\tilde\tau_{\text{WK}} (\bm t) := \tau_{\text{WK}} (\bm t/2)
\ee
is automatically a tau-function of the BKP hierarchy
due to \cite{al1}.
The following Schur Q-expansion formula was proposed by Mironov-Morozov \cite{mm}
(see also \cite{al2},
and see \cite{ly1} for a proof by Virasoro constraints):
\be
\tau_{\text{WK}} (\bm t)=
\sum_{\lambda\in DP} \Big(\frac{\hbar}{16} \Big)^{|\lambda|/3} \cdot
2^{-l(\lambda)} \frac{Q_\lambda(\delta_{k,1})Q_{2\lambda}(\delta_{k,3}/3)}
{Q_{2\lambda}(\delta_{k,1})}Q_\lambda(\bm t),
\ee
where $Q_\lambda(\delta_{k,1}$) means evaluating at the time $t_k = \delta_{k,1}$ for every $k$
in the Schur Q-function $Q_\lambda$.
This formula is inspired by the work \cite{diz},
see also \cite{jo}.

Then one is able to read off the BKP-affine coordinates for $\tilde\tau_{\text{WK}}$
by observing the coefficients of $Q_\lambda$ with $l(\lambda)\leq 2$.
For simplicity we take $\hbar=1$.
The evaluation $Q_\lambda(\delta_{k,1})$ is given by the hook-length type formula
(see e.g. \cite[(56)]{mm}):
\be
\label{eq-hooklength}
Q_\lambda(\delta_{k,1}) = \frac{2^{|\lambda|}}{\prod_{i=1}^{l(\lambda)} \lambda_i !}
\cdot \prod_{i<j} \frac{\lambda_i-\lambda_j}{\lambda_i+\lambda_j},
\ee
and $Q_\lambda(\delta_{k,3}/3)$ for $|\lambda|\leq 2$ is given by (see \cite[Theorem 3.1]{ly1}):
\begin{equation*}
\begin{split}
&Q_{(3m,3n)}(\delta_{k,3}/3) = \Big(\frac{2}{3}\Big)^{m+n} \cdot \frac{(m-n)}{(m+n)\cdot m!n!},
\qquad m>n\geq 0;\\
&Q_{(3m+1,3n+2)}(\delta_{k,3}/3) =
\Big(\frac{2}{3}\Big)^{m+n+1} \cdot \frac{2}{(m+n+1)\cdot m!n!},
\qquad m>n\geq 0;\\
&Q_{(3m+2,3n+1)}(\delta_{k,3}/3) =
-\Big(\frac{2}{3}\Big)^{m+n+1} \cdot \frac{2}{(m+n+1)\cdot m!n!},
\qquad m\geq n\geq 0;\\
\end{split}
\end{equation*}
and $Q_\lambda(\delta_{k,3}/3)=0$ if $|\lambda|\not\equiv 0 (\text{mod }3)$.
Thus the affine coordinates of the BKP tau-function $\tilde\tau_{\text{WK}}(\bm t)$ are given by:
\be
a_{0,3m}^{\text{WK}} = -a_{3m,0}^{\text{WK}}
=\frac{2\cdot(6m-1)!!}{4^{m+1}9^{m}\cdot (2m)!},
\qquad m>0,
\ee
and
\be
\begin{split}
&a_{3n,3m}^{\text{WK}} = \frac{(m-n)(6m-1)!!(6n-1)!!}{4^{m+n+1}9^{m+n}(m+n)(2m)!(2n)!},
\qquad n,m> 0;\\
&a_{3n+2,3m+1}^{\text{WK}} = -\frac{(6m+1)!!(6n+3)!!}{4^{m+n+2} 9^{m+n+1}(m+n+1)(2m)!(2n+1)!},
\qquad n, m\geq 0;\\
&a_{3n+1,3m+2}^{\text{WK}} = \frac{(6n+1)!!(6m+3)!!}{4^{m+n+2} 9^{m+n+1}(m+n+1)(2n)!(2m+1)!},
\qquad n, m\geq 0,
\end{split}
\ee
and
\be
a_{m,n}^{\text{WK}} =0,\qquad \text{if $m+n\not\equiv 0 (\text{mod }3)$}.
\ee
Here we use the conventions $(-1)!!=1$ and $0!=1$.
Let $A^{\text{WK}}(w,z)$ be the generating series of $\{a_{n,m}^{\text{WK}}\}$:
\begin{equation*}
A^{\text{WK}}(w,z)= \sum_{n,m>0} (-1)^{m+n+1} a_{n,m}^{\text{WK}} w^{-n} z^{-m}
-\half \sum_{n>0}(-1)^n a_{n,0}^{\text{WK}} (w^{-n}-z^{-n}).
\end{equation*}
The following are first a few terms of $A^{\text{WK}}(w,z)$:
\begin{equation*}
\begin{split}
&A^{\text{WK}}(w,z)=
\frac{5}{96} z^{-3} + \frac{1}{48} z^{-2} w^{-1}
- \frac{1}{48} z^{-1} w^{-2} - \frac{5}{96} w^{-3}
-\frac{385}{4608} z^{-6}\\
&\qquad -\frac{35}{2304} z^{-5} w^{-1}
+\frac{35}{2304} z^{-4} w^{-2}
-\frac{35}{2304} z^{-2} w^{-4}
+\frac{35}{2304} z^{-1} w^{-5}\\
&\qquad  +\frac{385}{4608} w^{-6}
+\frac{85085}{331776 }z^{-9}
+\frac{5005}{165888 }z^{-8}w^{-1}
-\frac{5005}{165888 }z^{-7}w^{-2}\\
&\qquad +\frac{1925}{331776 }z^{-6}w^{-3}
+\frac{1225}{82944 }z^{-5}w^{-4}
-\frac{1225}{82944 }z^{-4}w^{-5}
-\frac{1925}{331776 }z^{-3}w^{-6}\\
&\qquad +\frac{5005}{165888 }z^{-2}w^{-7}
-\frac{5005}{165888 }z^{-1}w^{-8}
-\frac{85085}{331776}w^{-9}
+\cdots.
\end{split}
\end{equation*}

\subsection{A formula for the generating series $A^{\text{WK}}(w,z)$}

It is not easy to find a simple formula for the generating series $A^{\text{WK}}(w,z)$
directly using the above expressions of $\{a_{n,m}^{\text{WK}}\}$.
However,
one can do this using Theorem \ref{thm-KPBKP}.
The explicit formulas of
the KP-affine coordinates $\{a_{n,m}^{\text{Zhou}}\}_{n,m\geq 0}$ of $\tau_{\text{WK}}$
was given in \cite{zhou3},
and their generating series are (\cite[(282)]{zhou1}):
\be
\label{eq-KP-WK-gen}
\sum_{n,m\geq 0}a_{n,m}^{\text{Zhou}} x^{-n-1}y^{-m-1} =
\frac{1}{y-x} + \frac{a(y)b(-x)-a(-x)b(y)}{y^2-x^2},
\ee
where $a(z),b(z)$ are the Faber-Zagier series \cite{ppz}:
\be
a(z)= \sum_{m=0}^\infty \frac{(6m-1)!!}{36^m\cdot (2m)!}  z^{-3m},\quad
b(z)= -\sum_{m=0}^\infty \frac{(6m-1)!!}{36^m \cdot (2m)!}\cdot \frac{6m+1}{6m-1}  z^{-3m+1}.
\ee
They are the first two basis vectors
of the point in Sato-Grassmannian corresponding to $\tau_{\text{WK}}$.
Thus by Theorem \ref{thm-KPBKP} one has:
\begin{Proposition}
The generating series of the BKP-affine coordinates $a_{m,n}^{\text{WK}}$ of $\tilde\tau_{\text{WK}} (\bm t)$ are given by:
\be
\label{eq-amngen-WK}
\begin{split}
&A^{\text{WK}}(w,z)
= \frac{w-z+a(-z)b(-w)-a(-w)b(-z)}{4(w+z)},\\
&\wA^{\text{WK}}(w,z)
= \frac{a(-z)b(-w)-a(-w)b(-z)}{4(w+z)}.
\end{split}
\ee
\end{Proposition}

Now one can plug $A^{\text{WK}}(w,z)$ and $\wA^{\text{WK}}(w,z)$ into
Theorem \ref{thm-main-conn} to obtain numerical data and
formulas for the connected $n$-point functions.
The following are first a few terms of the free energy:
\begin{equation*}
\begin{split}
&\log\tilde\tau_{\text{WK}} (\bm t) =
\big(\frac{t_3}{16} +\frac{105}{256}t_9 + \frac{25025}{2048} t_{15}
+\frac{56581525}{65536}t_{21} + \frac{58561878375}{524288}t_{27} +\cdots\big)\\
&\quad
+\big(\frac{5}{32}t_1t_5 +\frac{3}{64}t_3^2 +\frac{1155}{512}t_1t_{11}
+\frac{945}{512}t_3t_9 +\frac{1015}{512}t_5t_7 + \frac{425425}{4096}t_1t_{17}
+\cdots\big)\\
&\quad
+\big(\frac{t_1^3}{48} +\frac{35}{128}t_1^2t_7 +\frac{15}{32}t_1t_3t_5 +\frac{3}{64} t_3^3
+\frac{15015}{2048}t_1^2t_{13} +\frac{3465}{256}t_1t_3t_{11}+\cdots \big) +\cdots.
\end{split}
\end{equation*}
The original free energy of the Witten-Kontsevich tau-function
is recovered from $\log\tilde\tau_{\text{WK}}(\bm t)$
by a rescaling $t_i \mapsto 2t_i$ for every $i$.

\subsection{Affine coordinates of Br\'ezin-Gross-Witten tau-function}

The Br\'ezin-Gross-Witten (BGW) tau-function $\tau_{\text{BGW}}(\bm t)$ was introduced in
the study of lattice gauge theory \cite{bg, gw},
and it conjecturally describes the intersection numbers of certain classes on $\Mbar_{g,n}$,
see Norbury \cite{no}.
It is known that $\tau_{\text{BGW}}(\bm t)$ is a tau-function of the KdV hierarchy \cite{mms},
thus
\be
\tilde\tau_{\text{BGW}}(\bm t):= \tau_{\text{BGW}}(\bm t/2)
\ee
is a tau-function of the BKP hierarchy.
The following Schur Q-function expansion was conjectured in \cite{al2} and proved in \cite{ly2, al3}
by two different methods:
\be
\tau_{\text{BGW}} (\bm t)=
\sum_{\lambda\in DP} \Big(\frac{\hbar}{16} \Big)^{|\lambda|} \cdot
2^{-l(\lambda)} \frac{Q_\lambda(\delta_{k,1})^3}
{Q_{2\lambda}(\delta_{k,1})^2}Q_\lambda(\bm t).
\ee
For simplicity we take $\hbar=1$.
Using \eqref{eq-hooklength},
one may find that:
\be
\begin{split}
& a_{0,n}^{\text{BGW}} = - a_{n,0}^{\text{BGW}}
= \frac{\big((2n-1)!! \big)^2}{2^{3n+1}\cdot n!},
\qquad n>0;\\
& a_{n,m}^{\text{BGW}} = -a_{m,n}^{\text{BGW}}
= \frac{(m-n)\cdot \big((2m-1)!!(2n-1)!!\big)^2}{2^{3(m+n)+2}\cdot (m+n)\cdot m!n!},
\qquad m,n>0.
\end{split}
\ee
Now let
\begin{equation*}
\begin{split}
&A^{\text{BGW}}(w,z)= \sum_{n,m>0} (-1)^{m+n+1} a_{n,m}^{\text{BGW}} w^{-n} z^{-m}
-\half \sum_{n>0}(-1)^n a_{n,0}^{\text{BGW}} (w^{-n}-z^{-n}),\\
&\wA^{\text{BGW}}(w,z) = A^{\text{BGW}}(w,z)- \frac{w-z}{4(w+z)},
\end{split}
\end{equation*}
The following are first a few terms of $A^{\text{BGW}}(w,z)$:
\begin{equation*}
\begin{split}
&A^{\text{BGW}}(w,z)=
-\frac{1}{32} w^{-1} +\frac{1}{32} z^{-1}
+ \frac{9}{512} w^{-2} -\frac{9}{512} z^{-2}
-\frac{75}{4096} w^{-3}\\
&\quad - \frac{3}{4096} w^{-2}z^{-1}
+ \frac{3}{4096} w^{-1}z^{-2} +\frac{75}{4096} z^{-3}
+ \frac{3675}{131072} w^{-4} + \frac{75}{65536} w^{-3}z^{-1}\\
&\quad
- \frac{75}{65536} w^{-1}z^{-3} - \frac{3675}{131072} z^{-4}
- \frac{ 59535}{1048576} w^{-5} - \frac{2205}{1048576} w^{-4}z^{-1}\\
&\quad
-\frac{135}{524288} w^{-3} z^{-2} + \frac{135}{524288} w^{-2} z^{-3}
+\frac{2205}{1048576} w^{-1}z^{-4} + \frac{ 59535}{1048576} z^{-5}
+\cdots.
\end{split}
\end{equation*}

\subsection{A formula for the generating series $A^{\text{BGW}}(w,z)$}

In the case of the BGW tau-function,
we can also find a simple formula for the generating series
following the discussions in \S \ref{sec-phi12}.
The first two basis vectors of the point associated to the tau-function
$\tau_{\text{BGW}}$ in the Sato-Grassmannian are
(see Alexandrov \cite{al4}):
\be
\label{eq-al-bgwbasis}
\begin{split}
&\Phi_1^{\text{BGW}} (z) =1+\sum_{k=1}^\infty \frac{\big((2k-1)!!\big)^2}{8^k\cdot k!} z^{-k},\\
& \Phi_2^{\text{BGW}} (z) = z- \sum_{k= 0}^\infty \frac{(2k-1)!!(2k+3)!!}{8^{k+1}\cdot (k+1)!}z^{-k}.
\end{split}
\ee
(Notice here $\Phi_1^{\text{BGW}}(z),\Phi_2^{\text{BGW}}(z)$ differ from those in \cite{al4}
by a rescaling $z\mapsto 2z$,
since we're picking the notations in \cite{zhou1} which is slightly different from that in \cite{al4}.)
These two vectors are related by a Kac-Schwarz operator \cite{al4}:
\be
\label{eq-KS-bgw}
\Phi_2(z) = (z\frac{\pd}{\pd z} +z -\half) \Phi_1(z).
\ee
Denote:
\begin{equation*}
a_k = \frac{\big((2k-1)!!\big)^2}{8^k\cdot k!},\qquad
b_k = -\frac{(2k-3)!!(2k+1)!!}{8^{k}\cdot k!},\qquad
k\geq 1,
\end{equation*}
and denote:
\begin{equation*}
G^{\text{BGW}}(z)= \left[
\begin{array}{cc}
1+\sum_{n\geq 1} a_{2n}z^{-n} & \sum_{n\geq 0} b_{2n+1}z^{-n}\\
\sum_{n\geq 1}a_{2n-1} z^{-n} & 1+\sum_{n\geq 1} b_{2n} z^{-n}
\end{array}
\right].
\end{equation*}

\begin{Lemma}
We have $\det G^{\text{BGW}} (z) =1$.
\end{Lemma}
\begin{proof}
It is clear that:
\begin{equation*}
G^{\text{BGW}}(z)= \left[
\begin{array}{cc}
(\Phi_1^{\text{BGW}}(x)+\Phi_1^{\text{BGW}}(-x)\big)/2 &
(\Phi_2^{\text{BGW}}(x)-\Phi_2^{\text{BGW}}(-x)\big)/(2x) \\
(\Phi_1^{\text{BGW}}(x)-\Phi_1^{\text{BGW}}(-x)\big)/2 &
(\Phi_2^{\text{BGW}}(x)+\Phi_2^{\text{BGW}}(-x)\big)/(2x)
\end{array}
\right]
\end{equation*}
where $x:=z^\half$, thus we only need to check the following identity:
\be
\Phi_1^{\text{BGW}}(x) \Phi_2^{\text{BGW}}(-x)+\Phi_1^{\text{BGW}}(-x) \Phi_2^{\text{BGW}}(x) =2x.
\ee
Denote
\begin{equation*}
\Psi(x):= \frac{1}{2x} \big(
\Phi_1^{\text{BGW}}(x) \Phi_2^{\text{BGW}}(-x)+\Phi_1^{\text{BGW}}(-x) \Phi_2^{\text{BGW}}(x)
\big),
\end{equation*}
then by \eqref{eq-KS-bgw} one can compute:
\begin{equation*}
\begin{split}
\frac{d}{dx} \Psi(x) =&
\half \Phi_1(-x)^{\text{BGW}}
\Big( ( \Phi_1^{\text{BGW}})'' (x)
+2 (\Phi_1^{\text{BGW}}) '(x)\Big)\\
& -  \half \Phi_1(x)^{\text{BGW}}
\Big( ( \Phi_1^{\text{BGW}})'' (-x)
+2 (\Phi_1^{\text{BGW}}) '(-x)\Big).
\end{split}
\end{equation*}
Using the explicit expressions \eqref{eq-al-bgwbasis} one can may directly check that
\begin{equation*}
( \Phi_1^{\text{BGW}})'' (x) +\frac{1}{4x^2} \Phi_1^{\text{BGW}}(x)
+2 ( \Phi_1^{\text{BGW}})' (x) =0,
\end{equation*}
then $\frac{d}{dx}\Psi(x) =0$.
Thus $\Psi(x)$ is a constant,
and one easily finds that it is $2$.
\end{proof}

Thus by Proposition \ref{prop-phi12} we know that:
\begin{Proposition}
The generating series of the BKP-affine coordinates $a_{m,n}^{\text{BGW}}$ of $\tilde\tau_{\text{BGW}} (\bm t)$ are given by:
\be
\begin{split}
&A^{\text{BGW}}(w,z)
= \frac{w-z+\Phi_1^{\text{BGW}}(-z)\Phi_2^{\text{BGW}}(-w)-\Phi_1^{\text{BGW}}(-w)\Phi_2^{\text{BGW}}(-z)}{4(w+z)},\\
&\wA^{\text{BGW}}(w,z)
= \frac{\Phi_1^{\text{BGW}}(-z)\Phi_2^{\text{BGW}}(-w)-\Phi_1^{\text{BGW}}(-w)\Phi_2^{\text{BGW}}(-z)}{4(w+z)},
\end{split}
\ee
where $\Phi_1^{\text{BGW}}(z),\Phi_2^{\text{BGW}}(z)$ are given by \eqref{eq-al-bgwbasis}.
\end{Proposition}

Then the connected $n$-point functions can be computed by Theorem \ref{thm-main-conn}.
The first a few terms of the free energy are:
\begin{equation*}
\begin{split}
&\log\tilde\tau_{\text{BGW}} (\bm t) =
\big(
\frac{t_1}{16} + \frac{9}{256}t_3 + \frac{225}{2048} t_5 + \frac{55125}{65536} t_7 + \frac{6251175}{524288} t_9
+\cdots\big)\\
&\qquad\quad
+\big( \frac{t_1^2}{64}+\frac{27}{512} t_1t_3 + \frac{1125}{4096}t_1t_5 + \frac{567}{4096} t_3^2
+ \frac{385875}{131072} t_1t_7
+\cdots\big)\\
&\qquad\quad
+\big( \frac{t_1^3}{192} +\frac{27}{512} t_1^2t_3 + \frac{3375}{8192} t_1^2t_5 + \frac{1701}{4096} t_1t_3^2
+\cdots \big) +\cdots.
\end{split}
\end{equation*}
And $\log\tau_{\text{BGW}}(\bm t)$ is obtained by a rescaling $t_i\mapsto 2t_i$.

\vspace{.2in}

{\em Acknowledgements}.
We thank the anonymous referees for suggestions.
Z.W. thanks Professor Jian Zhou for helpful discussions,
and Professor Huijun Fan for encouragement.
C.Y. thanks Professor Xiaobo Liu for patient guidance.

%\section{Conflict of Interest Statement and Data Availability Statement}
%On behalf of all authors, the corresponding author states that there is no conflict of interest.

%All data generated or analysed during this study are included in this published article.

\end{document}